%% file: OLAP_v1.tex
\documentclass[12pt&reqno]{amsart}
\usepackage{csvsimple}
\usepackage{booktabs}
\input{header_main}

\begin{document}

  \title[The One-step Laplace Approximation]{Laplace approximation for Bayesian variable selection via Le Cam's one-step procedure}
  \thanks{This work is partially supported by the NSF grant DMS 2210664}
  
  \author{Tianrui Hou}\thanks{Tianrui Hou: Boston University, 675 Commonwealth Av., Boston 02215 MA, United States. {\em E-mail address:} tianruih@bu.edu}
  \author{Liwei Wang}\thanks{Liwei Wang: Boston University, 675 Commonwealth Av., Boston 02215 MA, United States. {\em E-mail address:} wlwfoo@bu.edu}
  \author{Yves Atchad\'e}\thanks{ Y. Atchad\'e: Boston University, 675 Commonwealth Av., Boston 02215 MA, United States. {\em E-mail address:} atchade@bu.edu}
  
  \subjclass[2010]{62F15, 62Jxx}
  
  \keywords{Regression models, Variable selection, High-dimensional Bayesian inference, Laplace approximation,  Markov Chain Monte Carlo mixing times}
  
  \maketitle

  \begin{center} (May 2024) \end{center}
  
  \begin{abstract}
Variable selection in high-dimensional spaces is a pervasive challenge in contemporary scientific exploration and decision-making. However, existing approaches  that are known to enjoy strong statistical guarantees often struggle to cope with the computational demands arising from the high dimensionality. To address this issue, we propose a novel Laplace approximation method based on Le Cam's one-step procedure (\textsf{OLAP}), designed to effectively tackles the computational burden. Under some classical high-dimensional assumptions we show that \textsf{OLAP} is a statistically consistent variable selection procedure. Furthermore, we show that the approach produces a posterior distribution that can be explored in polynomial time using a simple Gibbs sampling algorithm. Toward that polynomial complexity result, we also made some general, noteworthy contributions to the mixing time analysis of Markov chains. We illustrate the method using logistic and Poisson regression models applied to simulated and real data examples.
\end{abstract}

  \section{Introduction}
  Variable selection and model selection methods are crucial tools in the advancement on many data-driven scientific problems (\cite{ohara:sill:09,fan2010selective}). However existing methods that are known to be statistically consistent in high-dimensional regimes are computationally expensive to deploy (\cite{chen:chen:08,luo:chen:13,barber2016laplace}). As a result these methods are often implemented using greedy searches that may lack convergence guarantees. 
We develop in this work a new Laplace approximation that yields a variable selection procedure that is both consistent, and computationally inexpensive to implement. Although our methodology applies more generally, we focus the presentation on a class of generalized linear regression (GLM) models with known dispersion parameter. Specifically, we assume that we have a data set $\D\eqdef\{(y_i,{\bf x}_i),\;1\leq i\leq n\}$, with random response $y_i\in\rset$, and nonrandom explanatory variable ${\bf x}_i\in\rset^p$ collected from $n$ independent units. We consider a GLM model where, up to additive constant that we ignore, the log-likelihood function is given by
  \begin{equation}\label{ll:fun}
    \ell(\theta;\D) \eqdef \sum_{i=1}^n y_i\pscal{\theta}{{\bf x}_i} -\psi\left(\pscal{\theta}{{\bf x}_i}\right),\;\;\theta\in\rset^p,\end{equation}
  for some known function $\psi:\rset\to\rset$, and where $\theta\in\rset^p$ is the parameter of interest. We assume throughout that the model is well-specified, with a  true value of the parameter $\theta_\star\in\rset^p$ that satisfies  
  \[\PE\left(y_i - \psi'(\pscal{{\bf x}_i}{\theta_\star})\right) =0,\]
where $\psi'$ denotes the derivative of $\psi$. We consider the classical high-dimensional scenario where $p$ is potentially much larger than the sample size $n$, but $\theta_\star$ is believed to be sparse. The variable selection problem in this context corresponds to finding the support of $\theta_\star$.  We take a Bayesian approach and introduce an additional variable $\delta\in\Delta\eqdef\{0,1\}^p$ that encodes the support of $\theta$, with prior distribution 
  \begin{equation}\label{prior:1}
    \pi(\delta) \propto \left(\frac{1}{p^\mathsf{u}}\right)^{\|\delta\|_0},\;\;\; \delta\in\Delta,\end{equation}
  for some user-defined parameter $\mathsf{u}>0$. Furthermore in our prior specification, we assume that given $\delta$ the components of $\theta$ are independent, and $\theta_j\vert\{\delta_j=1\}$ follows a Gaussian prior $\textbf{N}(0,1)$, whereas $\theta_j\vert\{\delta_j=0\}$ follows a degenerate point-mass distribution with full mass at $0$. The variable selection problem therefore boils down to  sampling from (or computing the mode of) the marginal posterior distribution of $\delta$ given the data $\D$ given by
  \begin{multline}\label{post:Pi}
    \Pi(\delta\vert \D) \propto \left(\frac{1}{p^\mathsf{u}}\sqrt{\frac{1}{2\pi}}\right)^{\|\delta\|_0} \;\int_{\rset^{\|\delta\|_0}}\; \exp\left(-\frac{1}{2}\|w\|_2^2 + \ell^\delta(w;\D)\right)\rmd w,\;\\
    \;\;\;\mbox{ where }\;\;\; \ell^\delta(w;\D) \eqdef  \ell((w,0)_\delta;\D),\end{multline}
  and where the notation $(w,0)_\delta$ denotes the vector of $\rset^p$ obtained from $w$ by adding $0$ at components where $\delta_j = 0$ (see Section \ref{sec:notations} below for a more precise definition).  When the log-likelihood function $\ell$ is quadratic,  the integrals in (\ref{post:Pi}) have an explicit form, and it is possible to sample from (\ref{post:Pi}) using off-the-shelves MCMC methods on $\Delta$. This approach -- sometimes with slightly different priors -- has been a cornerstone of Bayesian variable selection for several decades (\cite{george1993variable,ohara:sill:09,guan:stephens:11,yang:etal:15}). Further MCMC advancements to deal with this type of discrete posterior distributions have appeared recently in the literature, although not considered here (\cite{zanella:roberts:19,griffin:etal:20,chang:quan:2024}).
  
It is worth noting that this marginalization strategy is generally not applicable  beyond the quadratic case. One common solution is data augmentation that samples jointly a model $\delta$ together with its corresponding parameter, using trans-dimensional MCMC algorithms (\cite{green2003trans,lamnisos:etal:09}). However, this strategy is known to scale poorly, particularly because trans-dimensional MCMC algorithms are difficult to tune (\cite{brooksetal03}) and their sampling complexity is poorly understood. A related idea developed in (\cite{carlin:chib:95,AB:19}), consists in complementing the parameter set of each model $\delta$ with additional variables drawn from the so-called pseudo-prior distribution.  In the same vein, several carefully constructed priors that are absolutely continuous with respect to the Lebesgue measure have also been proposed and studied in the literature (\cite{polson:james:11,narissety:he:14,rockova:george:18}; see \cite{bhadra:etal:19} for a review). These different approaches circumvent the trans-dimensional nature of the posterior distribution, but requires additional tuning of the prior distribution. Furthermore, an important benefit of the marginalization strategy in (\ref{post:Pi}) is lost: integrating out $\theta$ and sampling only from $\delta$ (the so-called collapsed Gibbs sampler) typically improves MCMC mixing (see e.g \cite{liu:94}).
  
A popular solution in cases where the integrals in (\ref{post:Pi}) are not available in closed form is the use of Laplace approximation. The modes of the resulting posterior approximation are known to be equivalent to the BIC or e-BIC (\cite{original_bic,chen:chen:08}). In (\cite{barber2016laplace}) it is shown that under some regularity conditions, and for a class of generalized linear models, replacing the integrals in (\ref{post:Pi}) by their Laplace approximations still produces consistent variable selection. Specifically the method replaces the function $\ell^\delta(\cdot;\D)$ by its second order Taylor expansion around the estimator
  \begin{equation}\label{hat:theta}
    \hat\theta^{\delta} \eqdef \argmax_{w\in\rset^{\|\delta\|_0}}\left[-\frac{1}{2}\|w\|_2^2 + \ell^\delta(w;\D)\right].\end{equation}
  The major limitation of this approach is that the Laplace approximation requires the computation of the estimators $\hat\theta^{\delta}$ (as well as the Hessian matrix of $\ell^\delta$ at $\hat\theta^\delta$) for each $\delta$, which is typically computationally expensive. In particular in the GLM considered in this work, these estimators are not available in closed form and requires the use of a numerical solver. 

In the frequentist realm, the high-dimensional variable selection problem in a GLM model with likelihood such as (\ref{ll:fun}) is often framed as the best-subset selection problem that solves the minimization problem
\[\min -\ell(\theta;\D) \;\; \mbox{ subject to }\;\; \|\theta\|_0\leq s,\]
or some relaxation thereof, for a specified sparsity level $s$ (\cite{beck:eldar:13,hastie:etal:15}). These estimators have seen  renewed interest due to recent optimization advances to tackle an otherwise NP-hard problem (\cite{bertsimas:etal:2016,hazimeh:mazumder:20}). One drawback of this approach is the limited understanding of the statistical properties of the resulting estimators and the lack of effective and adaptive methods for selecting the sparsity level $s$.

  \subsection{Main contributions}
In this work we revisit the Laplace approximation method outlined above.  We circumvent the computational challenge of the method by introducing a simple variant of the Laplace approximation based on Le Cam's one-step procedure. The one-step procedure that goes back to (\cite{lecam:69}), is a well-known scheme to improve on a given estimator through a one step Newton-Raphson update. We refer the reader to (\cite{vdv:98}~Section 5.7) for some basic properties. The method has also seen a resurgence in popularity in recent years in the high-dimensional literature as a way to de-bias high-dimensional estimators (\cite{geer:etal:14,javanmard:etal:14,xia2020revisit}). Hence, starting from an initial estimator $\wtilde\theta$ that is computed once, we propose  fast approximations of the estimators $\hat\theta^{\delta}$ for any $\delta$, using Le Cam's one-step procedure. This leads to a new Laplace approximation that we call one-step Laplace approximation (\textsf{OLAP}). Under some basic high-dimension regularity conditions (basically restricted strong convexity, and limited coherence of Hessian matrices) we show that the method has a statistical performance comparable to the standard Laplace approximation with knowledge of the estimators $\hat\theta^{\delta}$, while being an order of magnitude faster.  We also found empirically that \textsf{OLAP} is statistically more accurate than other state-of-the-art methods such as Skinny-Gibbs (\cite{narisetty2018skinny}) and SparseVB (\cite{ray2020spike}) for variable selection in logistic regression.

We also analyze a simple Gibbs sampler to sample from the posterior distribution of \textsf{OLAP}. Under the same regularity assumptions mentioned above, we show that a plain Gibbs sampler applied to the posterior distribution of \textsf{OLAP} has a mixing time that scales polynomially with $(n,p)$. Toward the proof we also derive some general results on the mixing times of Markov Chain Monte Carlo methods that may be of independent interest. In particular we establish in Theorem \ref{lem:cheeger} a new generalization of Cheeger's inequality for Markov chains that connects the $\epsilon$-conductance of \cite{lovasz:simonovits93} and the approximate spectral gap of \cite{atchade:asg}.
  
  \subsection{Outline} 
  We end this introduction with some general notations that are employed throughout the paper. The one-step Laplace approximation of (\ref{post:Pi}) and some basic statistical and computational properties are derived in Section \ref{sec:olap}. Numerical illustrations of the methods are presented in Section \ref{sec:num}, and all the technical proofs are collected in the appendix. We close the paper with a brief summary in Section \ref{sec:conclusion}. Section \ref{sec:mix:mc} contains some new results on the mixing times of MCMC algorithms that we use in the proof of Theorem \ref{thm:mix}.
  
  \subsection{Notations}\label{sec:notations}
  We also collect here  our notations on sparse models.  Throughout our parameter space is $\rset^p$ equipped with its Euclidean norm $\|\cdot\|_2$ and inner product $\pscal{\cdot}{\cdot}$.  We also use $\|\cdot\|_0$ which counts the number of non-zero elements, and $\|\cdot\|_\infty$ which returns the largest absolute value. 
  
  We set $\Delta\eqdef\{0,1\}^p$. Elements of $\Delta$ are called sparsity structures, or supports. For $\delta,\delta'\in\Delta$, we define $\min(\delta,\delta')$ as the component-wise minimum of the vectors $\delta$ and $\delta'$. And we write $\delta\subseteq \delta'$ if $\min(\delta,\delta')=\delta$. Equivalently, $\delta\subseteq\delta'$ if $\delta_j\leq \delta_j'$ for all $1\leq j\leq p$, and we write $\delta\supseteq\delta'$ if $\delta'\subseteq\delta$. 
  
  Given $\delta\in\Delta$, and $\theta\in\rset^p$, we write $\theta_\delta$ to denote the  component-wise product of $\theta$ and $\delta$, and $\delta^c\eqdef 1-\delta$. Assuming $\|\delta\|_0=s$, we will use the notation $[\theta]_\delta$ to denote the vector $(\theta_{j_1},\ldots,\theta_{j_s})$, where $j_i$ is the $i$-th component of $\delta$ that is non-zero.  Conversely, for $w\in\rset^{s}$, we define $(w,0)_\delta$ as the element of $\rset^p$ such that $[(w,0)_\delta]_\delta = w$. At times we will abuse the notation and use $\theta_\delta$ and $[\theta]_\delta$ interchangeably.

  \section{The one-step Laplace approximation}\label{sec:olap}
  With $\ell^\delta$ as in (\ref{post:Pi}), we define
  \[\bar\ell^\delta(w;\D) \eqdef \ell^\delta(w;\D) -\frac{1}{2}\|w\|_2^2,\;\; w\in\rset^{\|\delta\|_0}.\]
  The Laplace approximation of the posterior distribution $\Pi(\cdot\vert\D)$ is formed by replacing $\bar\ell^\delta$ in the posterior (\ref{post:Pi}) by its second order Taylor approximation around the mode $\hat\theta^{\delta}$,  where  $\hat\theta^{\delta}$ is  as defined  in (\ref{hat:theta}). The resulting Gaussian integral has an explicit form, which leads to the approximation of (\ref{post:Pi})  given by
  \begin{equation}\label{basic:lappace:approx}
    \hat\Pi(\delta\vert \D) \propto \left(\frac{1}{p^\mathsf{u}}\right)^{\|\delta\|_0}  \frac{e^{\bar\ell^\delta(\hat\theta^{\delta};\D)}}{\sqrt{\det\left(\hat\H^\delta\right)}},\;\;\mbox{ where }\;\;\hat\H^\delta\eqdef -\nabla^{(2)}\bar\ell^\delta(w;\D)\vert_{w=\hat\theta^{\delta}}. \end{equation}
  \cite{barber2016laplace} studied the statistical properties of $\hat\Pi$ and showed that it is consistent (in the sense that it puts high probability on the true model) under some regularity conditions, as $n,p\to\infty$. However as observed above, in general the estimator $\hat\theta^{\delta}$ is computationally expensive to obtain, making the approach difficult to use in large-scale applications.

  We propose \textsf{OLAP}, a simple modification based on Le Cam's one-step device. The basic idea is as follows. For any $v_0\in\rset^{\|\delta\|_0}$, using a second order Taylor approximation of the function $\bar\ell^\delta(\cdot;\D)$ around $v_0$, we can approximate the integral $\int_{\rset^{\|\delta\|_0}}\; e^{\bar\ell^\delta(w;\D)}\rmd w$ by
  \begin{equation}\label{laplace:approx}
    \frac{(2\pi)^{\frac{\|\delta\|_0}{2}}}{\sqrt{\det(\H^{v_0})}} \exp\left(\bar\ell(v_0;\D) + \frac{1}{2}(\G^{v_0})^{\texttt{T}}(\H^{v_0})^{-1} \G^{v_0}\right).\end{equation}
  In the above expression, $\G^{v_0}\eqdef \nabla\bar\ell^\delta(w;\D)\vert_{w=v_0}$, and $\H^{v_0} \eqdef -\nabla^{(2)}\bar \ell^\delta(w;\D)\vert_{w=v_0}$. The approximation performs at its best when $v_0$ is near $\hat\theta^{\delta}$, the mode of $\bar\ell^\delta$, and this is what the standard Laplace approximation does. Let $\wtilde\theta\in\rset^p$ be some initial estimator of $\theta_\star$. The estimator $\wtilde\theta$ is computed only once, and for any $\delta\in\Delta$, we can take $\wtilde\theta^{\delta}\eqdef[\wtilde\theta]_\delta$ (where $[\wtilde\theta]_\delta$ is the sub-vector of $\tilde\theta$ that collects the components of $\tilde\theta$ for which $\delta_j=1$) as an initial approximation of $\hat\theta^{\delta}$, that we then improve upon using the one-step procedure. Hence our proposed approximation of $\hat\theta^{\delta}$ is 
  \begin{equation}\label{theta:check}
    \check\theta^{\delta}\eqdef   \tilde\theta^{\delta} + (\tilde\H^\delta)^{-1}\tilde \G^\delta,
  \end{equation}
  where
  \begin{equation}\label{def:GH}
    \tilde \G^\delta\eqdef \nabla\bar\ell^\delta(w;\D)\vert_{w=\tilde\theta^{\delta}},\;\;\;\mbox{ and }\;\; \tilde \H^\delta\eqdef -\nabla^{(2)}\bar \ell^\delta(w;\D)\vert_{w=\tilde\theta^{\delta}}.\end{equation}
  The estimator $\check\theta^{\delta}$ is obtained from $\tilde\theta^{\delta}$ by a single Newton-Raphson update. In statistical estimation theory, it is well-known that this single step is enough to improve any rate-optimal initial estimator into an efficient rate optimal estimator, by moving it toward the related M-estimator (\cite{vdv:98,brouste2021onestep}). Versions of the procedure have also appeared in the recent statistics literature to de-bias and otherwise improve upon classical high-dimensional estimators (\cite{geer:etal:14,javanmard:etal:14,xia2020revisit}). Hence we expect similar improvements in our context.
  
  Using $\check\theta^{\delta}$ and (\ref{laplace:approx}) leads to the approximation of $\Pi$ given by
  \begin{equation}\label{Pi:check:pre}
    \delta\mapsto\left(\frac{1}{p^\mathsf{u}}\right)^{\|\delta\|_0}\frac{\exp\left(\bar\ell^\delta(\check\theta^{\delta};\D)+ \frac{1}{2}(\check\G^\delta)^{\texttt{T}}(\check\H^\delta)^{-1}\check\G^\delta\right)}{\sqrt{\det\left(\check\H^\delta\right)}},\;\;\delta\in\Delta,
  \end{equation}
  where $\check \G^\delta$ and $\check\H^\delta$ are defined as in (\ref{def:GH}), but replacing $\tilde\theta^{\delta}$ by $\check\theta^{\delta}$.  We note that if instead of a single update, we iterate (\ref{theta:check}) a large number of times, then naturally $\check\theta^{\delta}\approx \hat\theta^{\delta}$, and $\check\G^\delta\approx 0$, and (\ref{Pi:check:pre}) becomes (\ref{basic:lappace:approx}). Hence we propose to drop the term $(\check\G^\delta)^{\texttt{T}}(\check\H^\delta)^{-1}\check\G^\delta$, as well as the log-determinant term in (\ref{Pi:check:pre}), leading to OLAP, our proposed Laplace approximation 
  \begin{equation}\label{Pi:check:2}
    \check\Pi(\delta\vert \D) \propto \left(\frac{1}{p^\mathsf{u}}\right)^{\|\delta\|_0}  e^{\bar\ell^\delta(\check\theta^{\delta};\D)},\;\;\delta\in\Delta.
  \end{equation}

We then solve the variable selection problem by sampling from the OLAP distribution (\ref{Pi:check:2}).

  \subsubsection{Initial estimator $\tilde\theta$}
Throughout we take the initial estimator $\tilde\theta\in\rset^p$ as either lasso, the ridge regression estimator, or more generally the elastic-net estimator defined as
  \[\argmin_{\theta\in\rset^p}\;\left[\sum_{i=1}^n -y_i\pscal{\theta}{{\bf x}_i} +\psi\left(\pscal{\theta}{{\bf x}_i}\right) + \lambda_1 \|\theta\|_1 + \frac{\lambda_2}{2}\|\theta\|_2^2 \right],\]
  for some well-tuned regularization parameters $\lambda_1\geq 0$, $\lambda_2\geq 0$. We refer the reader to \cite{hastie:etal:15} and the extensive literature on these high-dimensional estimators. 

  \subsection{Statistical properties} We recall that we have assumed that our data $\D=\{(y_i,{\bf x}_i),\;1\leq i\leq n\}$ is generated from a well-specified  GLM model, with  true value of the parameter $\theta_\star$ with support $\delta_\star$, and we set $s_\star\eqdef\|\theta_\star\|_0$. In this section we search for conditions under which the \textsf{OLAP} posterior distribution introduced in (\ref{Pi:check:2}) concentrates around $\delta_\star$. We start with the following general definition.  Let $\{\hat\beta^\delta,\;\delta\in\Delta\}$ be some arbitrary family of estimators, where for $\delta\in\Delta$, $\hat\beta^\delta\in\rset^{\|\delta\|_0}$. We say that the family $\{\hat\beta^\delta,\;\delta\in\Delta\}$ is variable selection consistent if there exists positive constants $c_1,c_2<\infty$ such that the following two properties holds. 
  \begin{enumerate}
    \item For all  $\delta\in\Delta$, and $\delta_0\subseteq\delta$, with $\min(\delta,\delta_\star)=\delta_0$, it holds
\begin{equation}\label{msc:eq1}
\bar\ell(\hat\beta^\delta;\D) \leq \bar\ell(\hat\beta^{\delta_0};\D)  + c_1 (\|\delta\|_0 - \|\delta_0\|_0)\log(p),
\end{equation}
    \item and for all  $\delta_0,\delta\in\Delta$ such that $\delta_0\subseteq\delta\subseteq\delta_\star$, it holds
\begin{equation}\label{msc:eq2}
\bar\ell(\hat\beta^{\delta};\D) \geq \bar\ell(\hat\beta^{\delta_0};\D)  + c_2 (\|\delta\|_0 - \|\delta_0\|_0)n.\end{equation}
  \end{enumerate}
  
  \begin{remark}
    This definition basically says that $\{\hat\beta^\delta,\;\delta\in\Delta\}$ is model selection consistent if increasing a model $\delta_0$ by adding a non-relevant variable increases the  log-likelihood $\bar\ell(\hat\beta^{\delta_0};\D)$ only by a factor at most $\log(p)$, whereas adding a relevant variable increases the log-likelihood by a factor $n$. The behavior depicted in this definition is generally the expected behavior of the (maximum) log-likelihood ratio. However establishing that this behavior prevails for a given model can sometimes be  very challenging, particularly in the high-dimensional setting.
  \end{remark}
  
  Our first result shows that if the one-step family $\{\check\theta^{\delta},\;\delta\in\Delta\}$ is variable selection consistent as defined above, then the OLAP distribution (\ref{Pi:check:2}) puts most of its probability mass on $\delta_\star$. For $j\geq 0$, we define
  \[\A_j\eqdef \left\{\delta\in\Delta:\; \delta\supseteq \delta_\star,\;\mbox{ and }\;\; \|\delta\|_0\leq s_\star +j\right\}.\]

  \begin{theorem}\label{thm:post:contr}
Assume $p\geq 2$. Suppose that the family of one-step estimators $\{\check\theta^\delta,\;\delta\in\Delta\}$ introduced in (\ref{theta:check}) is model selection consistent with constants $c_1,c_2$. Consider the posterior distribution (\ref{Pi:check:2}), and suppose that the sparsity parameter $\mathsf{u}$ satisfies
    \begin{equation}\label{cond:u}
      \mathsf{u}\geq 2(1 + c_1),
    \end{equation}
    and the sample size $n$ satisfies 
    \begin{equation}\label{cond:ss}
      c_2 n\geq 2(\mathsf{u}+1)\log(p).
    \end{equation}
    Then for all $j\geq 0$, we have
    \begin{equation}\label{eq:post:star}
      \check\Pi\left(\A_j\vert \D\right)  \geq 1 - 2\left(\frac{1}{p^{\frac{\mathsf{u}(j+1)}{2}}} + \frac{2}{e^{c_2n/4}}\right).
    \end{equation}
  \end{theorem}
  \begin{proof}
    See Section \ref{sec:proof:thm:post:contr}.
  \end{proof}
  
\begin{remark}
We note that the result holds true for $j=0$ and yields
\[\check\Pi\left(\delta_\star\vert \D\right)  \geq 1 - 2\left(\frac{1}{p^{\frac{\mathsf{u}}{2}}} + \frac{2}{e^{c_2n/4}}\right).\]
The theorem assumes that the family of one-step estimators $\{\check\theta^\delta,\;\delta\in\Delta\}$ is model selection consistent. We investigate this key assumption below. An important point to make here is that in fact any family of estimators that is model selection consistent as defined above yields a consistent estimation of $\delta_\star$. Specifically, it can be similarly shown that if a family of estimators $\{\hat\beta^\delta,\;\delta\in\Delta\}$ is model selection consistent then the distribution 
\[\delta\mapsto \exp\left(-\mathsf{u}\|\delta\|_0\log(p) + \bar\ell^\delta(\hat\beta^\delta;\D)\right),\;\;\delta\in\Delta,\]
satisfies the conclusion of the theorem. Of course, this property is of interest only if the estimators $\hat\beta^\delta$ are easy to compute.
\end{remark}

\subsection{MCMC sampling and mixing time}
  One of the main advantage of OLAP is that the distribution (\ref{Pi:check:2}) can be easily explored using a simple Gibbs sampler on the discrete space $\Delta$.  We note that  the conditional distribution of $\delta_j$ given $\delta_{-j}$ is the Bernoulli distribution $\textbf{Ber}(q_j(\delta))$, with probability of success given by
  \begin{equation}\label{cond:dist:eq:1}
    q_j(\delta) \eqdef \left(1 +\exp\left(\mathsf{u}\log(p) + \bar\ell^{\delta^{(j,0)}}(\check\theta_{\delta^{(j,0)}};\D) - \bar\ell^{\delta^{(j,1)}}(\check\theta_{\delta^{(j,1)}};\D) \right)\right)^{-1}.\end{equation}
  where $\delta^{(j,0)}$ (resp $\delta^{(j,1)}$) is equal to $\delta$ except at component $j$ where $\delta^{(j,0)}_j=0$ (resp. $\delta^{(j,1)}_j=1$). The resulting  Gibbs sampler is presented in Algorithm \ref{algo:1}.
  
\bigskip

  \begin{algorithm}[Gibbs sampler for OLAP]\label{algo:1}$\hrulefill$\\
Pick $\delta^{(0)}\in\Delta$ the initial state. Repeat the following steps for $k=0,\ldots$. Given $\delta^{(k)} = \delta\in\Delta$:
    \begin{enumerate}
      \item Set $\bar\delta = \delta$. Randomly and uniformly select an ordered subset $\mathsf{J}\subset\{1,\ldots,p\}$ of size $J$.
      \item Update $(\bar\delta_{\mathsf{J}_1},\ldots,\bar\delta_{\mathsf{J}_J})$  sequentially: for $\iota=1,\ldots, J$, draw $V_\iota\sim \textbf{Ber}(q_{\mathsf{J}_\iota}(\bar\delta))$, where $q_j$ is as in (\ref{cond:dist:eq:1}), and set $\bar \delta_{\mathsf{J}_\iota}=V_\iota$.
      \item   Set $\delta^{(k+1)} = \bar \delta$.
    \end{enumerate}
    \vspace{-0.6cm}
  \end{algorithm}
  $\hrulefill$
  \bigskip
  
  The cost of computing $q_j(\delta)$ is driven by the cost of forming the matrix $\tilde{\H}^\delta$ in (\ref{theta:check}), plus the cost of the Cholesky factorization needed to perform the Newton-Raphson update. Hence the computational cost of $q_j(\delta)$ is of order $n\|\delta\|_0^2 + \|\delta\|_0^3$. As a result, we see that the computation cost of the $k$-th iteration of Algorithm \ref{algo:1} is of order
  \[J\times \max\left(n\|\delta^{(k)}\|_0^2,\; \|\delta^{(k)}\|_0^3\right).\]
  
We analyze the mixing time of Algorithm \ref{algo:1}. We focus on the case $J=1$, where the resulting Markov chain is reversible and positive. In the case $J>1$, a similar analysis can be developed for the lazy version of the algorithm, but we will not pursue this. Assuming the data $\D$ fixed, let $\mathbb{P}(\cdot\vert \delta)$ denote the probability measure of the Markov chain generated by Algorithm \ref{algo:1} started from $\delta$.
  
  \begin{theorem}\label{thm:mix}
Assume $p\geq 2$. Let $\{\delta^{(k)},\;k\geq 0\}$ be the Markov chain generated by Algorithm \ref{algo:1} with $J=1$ and initial state $\delta^{(0)}$. Suppose that the one-step family $\{\check\theta^\delta,\;\delta\in\Delta\}$ is variable selection consistent with constants $c_1,c_2$ such that (\ref{cond:u}) and (\ref{cond:ss}) hold. Suppose also that there exists $\alpha$, $0\leq \alpha \leq \mathsf{u} (p-s_\star)/4$ such that $\delta^{(0)}$ satisfies 
    \begin{equation}\label{cond:init:dist}
      \check\Pi(\delta^{(0)}\vert \D) \geq \frac{1}{p^\alpha}.
    \end{equation}
Fix $\zeta_0\in (0,1)$. Then there exists an absolute constant $C$ such that if  the number of iterations satisfies
    \begin{equation}
      N\geq C \left(s_\star + \frac{4\alpha}{\mathsf{u}}\right)\left(\log\left(\frac{1}{\zeta_0}\right) + \alpha\log(p)\right)p,\end{equation}
    then it holds
    \[\|\PP(\delta^{(N)}=\cdot\; \vert \delta^{(0)}) -\check\Pi(\cdot)\|_\tv \leq \max\left(\zeta_0,\frac{4}{p^{\frac{\mathsf{u}}{4}}}\right).\]
  \end{theorem}
  \begin{proof}
    See Section \ref{sec:proof:thm:mix}.
  \end{proof}

The main conclusion of the theorem is that under the warm-start condition (\ref{cond:init:dist}), if the one-step family $\{\check\theta^\delta,\;\delta\in\Delta\}$ is variable selection consistent then Algorithm \ref{algo:1} is fast mixing. Importantly the theorem also highlights the importance of the initial distribution. To provide more insight on this last point, suppose that we slightly strengthen the definition of model selection consistency of $\{\check\theta^\delta,\;\delta\in\Delta\}$ by replacing (\ref{msc:eq1}) and (\ref{msc:eq2}) respectively by 
\begin{equation}\label{msc:eq12}
\bar\ell(\check\theta^{\delta_0};\D) \leq \bar\ell(\check\theta^\delta;\D) \leq \bar\ell(\check\theta^{\delta_0};\D)  + c_1 (\|\delta\|_0 - \|\delta_0\|_0)\log(p),\end{equation}
and
\begin{equation}\label{msc:eq22}
\bar\ell(\check\theta^\delta;\D) + c_3 (\|\delta_\star\|_0 - \|\delta\|_0)n \geq \bar\ell(\check\theta_{\delta_\star};\D) \geq \bar\ell(\check\theta^\delta;\D)  + c_2 (\|\delta_\star\|_0 - \|\delta\|_0)n.
\end{equation}
for some constant $c_3\geq c_2$. Now, suppose that we run Algorithm \ref{algo:1} from some initial state $\delta^{(0)}\supseteq \delta_\star$, with $\|\delta^{(0)}\|_0 = k_0$. In other words, the initial state has no false-negative. Then, noting that $\check\Pi(\delta_\star\vert\D)\geq 1/2$ for $n$ and $p$ large enough as a consequence of Theorem \ref{thm:post:contr},  we have:
\[\check\Pi(\delta^{(0)}\vert\D) = \check\Pi(\delta_\star\vert\D) \frac{\check\Pi(\delta^{(0)}\vert\D)}{\check\Pi(\delta_\star\vert\D)} \geq \frac{1}{2p^{\mathsf{u}(k_0-s_\star)}} \exp\left( \bar\ell(\check\theta^{\delta^{(0)}};\D)-\bar\ell(\check\theta^{\delta_\star};\D)\right).\]
By the first inequality of (\ref{msc:eq12}), $\bar\ell(\check\theta^{\delta^{(0)}};\D)-\bar\ell(\check\theta^{\delta_\star};\D)\geq 0$. Hence, we conclude that
\[\check\Pi(\delta^{(0)}\vert\D)  \geq\frac{1}{2p^{\mathsf{u}(k_0-s_\star)}}.\]
This inequality means that we can apply Theorem \ref{thm:mix} with $\alpha = \mathsf{u}(k_0 - s_\star)/2$, and it follows that Algorithm \ref{algo:1} has a mixing time that is at most $k_0^2\log(p) \times p$. Crucially in this case, the mixing time of the algorithm does not directly depend on the sample size $n$.

Consider now the case where the initial state $\delta^{(0)}$ contains some false negatives. Specifically, suppose that we start Algorithm \ref{algo:1} from $\delta^{(0)}$ such that $\|\delta^{(0)}\|_0 = k_0$, but $\delta_\star^{(0)}\eqdef \min(\delta^{(0)},\delta_\star)$ is a strict sub-model of $\delta_\star$, say $\|\delta_\star^{(0)}\|_0=s_0<s_\star$. In that case, using the definition of $\check\Pi$ in (\ref{Pi:check:2}), and $\check\Pi(\delta_\star\vert\D)\geq 1/2$,
\begin{multline*}
\check\Pi(\delta^{(0)}\vert\D) = \check\Pi(\delta_\star\vert\D) \frac{\check\Pi(\delta_\star^{(0)}\vert\D)}{\check\Pi(\delta_\star\vert\D)}\frac{\check\Pi(\delta^{(0)}\vert\D)}{\check\Pi(\delta_\star^{(0)}\vert\D)} \\
\geq \frac{1}{2p^{\mathsf{u}(k_0 - s_\star)}}\exp\left(\bar\ell(\check\theta^{\delta_\star^{(0)}};\D) - \bar\ell(\check\theta^{\delta_\star};\D) + \bar\ell(\check\theta^{\delta^{(0)}};\D) - \bar\ell(\check\theta^{\delta_\star^{(0)}};\D)\right).
 \end{multline*}
By the first inequality of (\ref{msc:eq12}), $\bar\ell(\check\theta^{\delta^{(0)}};\D) - \bar\ell(\check\theta^{\delta_\star^{(0)}};\D)\geq 0$, and the first inequality of (\ref{msc:eq22}) yields $\bar\ell(\check\theta^{\delta_\star^{(0)}};\D) - \bar\ell(\check\theta^{\delta_\star};\D) \geq -c_3 (s_\star - s_0) n$. Hence it follows that
\[\check\Pi(\delta^{(0)}\vert\D) \geq \frac{\exp\left(-c_3(s_\star - s_0)n\right)}{2p^{\mathsf{u}(k_0 - s_\star)}} \geq \frac{1}{2p^{\alpha}},\]
with $\alpha = \mathsf{u}(k_0-s_\star) + \frac{c_3(s_\star - s_0)n}{\log(p)}$. In that case, it follows from Theorem \ref{thm:mix} that Algorithm \ref{algo:1} has a mixing time that is at most 
\[\left(k_0 + \frac{(s_\star -s_0)n}{\log(p)}\right)^2\log(p) \times p.\]
Here, the mixing time is still polynomial in $(n,p)$, however the result suggests that the mixing time is negatively impacted by the sample size $n$. This is similar to the worst-case mixing time bound of $n s_0^2 p\log(p)$ obtained by \cite{yang:etal:15} for sampling from a version of (\ref{basic:lappace:approx}) for variable selection in a high-dimensional linear regression model, where, using their notations, $s_0$ denotes an upper-bound on $s_\star$. Although we have a worst dependence on $n$, our work improves on \cite{yang:etal:15} in two ways. Firstly, instead of a worst-case analysis, our result describes more precisely the effect of the initialization. Secondly, our analysis does not require an a-priori bound on $s_\star$. We stress however that our results (as well as \cite{yang:etal:15}) only provide upper bounds on the mixing times. It is possible that Algorithm \ref{algo:1} actually converges faster than what is described in Theorem \ref{thm:mix}.

We illustrate these findings with a Poisson regression simulation example. The data generating process is described in Section \ref{sec:num}. Here we set $p=1000$ and we vary $n\in\{1000,1500,2000\}$. The true model has $10$ relevant variables. For each $(n,p)$ we generate $50$ datasets, leading to $50$ \textsf{OLAP} distributions $\check\Pi(\cdot\vert\D)$. For each posterior we estimate the mixing time of Algorithm \ref{algo:1}  using the $L$-lag coupling method  of \cite{biswas2019estimating}, employing $30$ coupled chains. We compare the mixing times of Algorithm \ref{algo:1} when started from the null model, and a version started from $\delta^{(0)}$ that contains the true model plus $10$ false-positive. The boxplots of the distributions of the mixing times are given in Figure \ref{diff_initialization_logistic}. The results are consistent with our theory. This finding highlights the importance of good MCMC initialization (warm-start). In all the simulation below we initialize the algorithm from the support of the \textsf{lasso} estimator, which under mild conditions is known to contain the true support for sufficiently strong signal (\cite{meinshausenetal09}).

    \begin{figure}[H]
    \includegraphics[width=.49\textwidth]{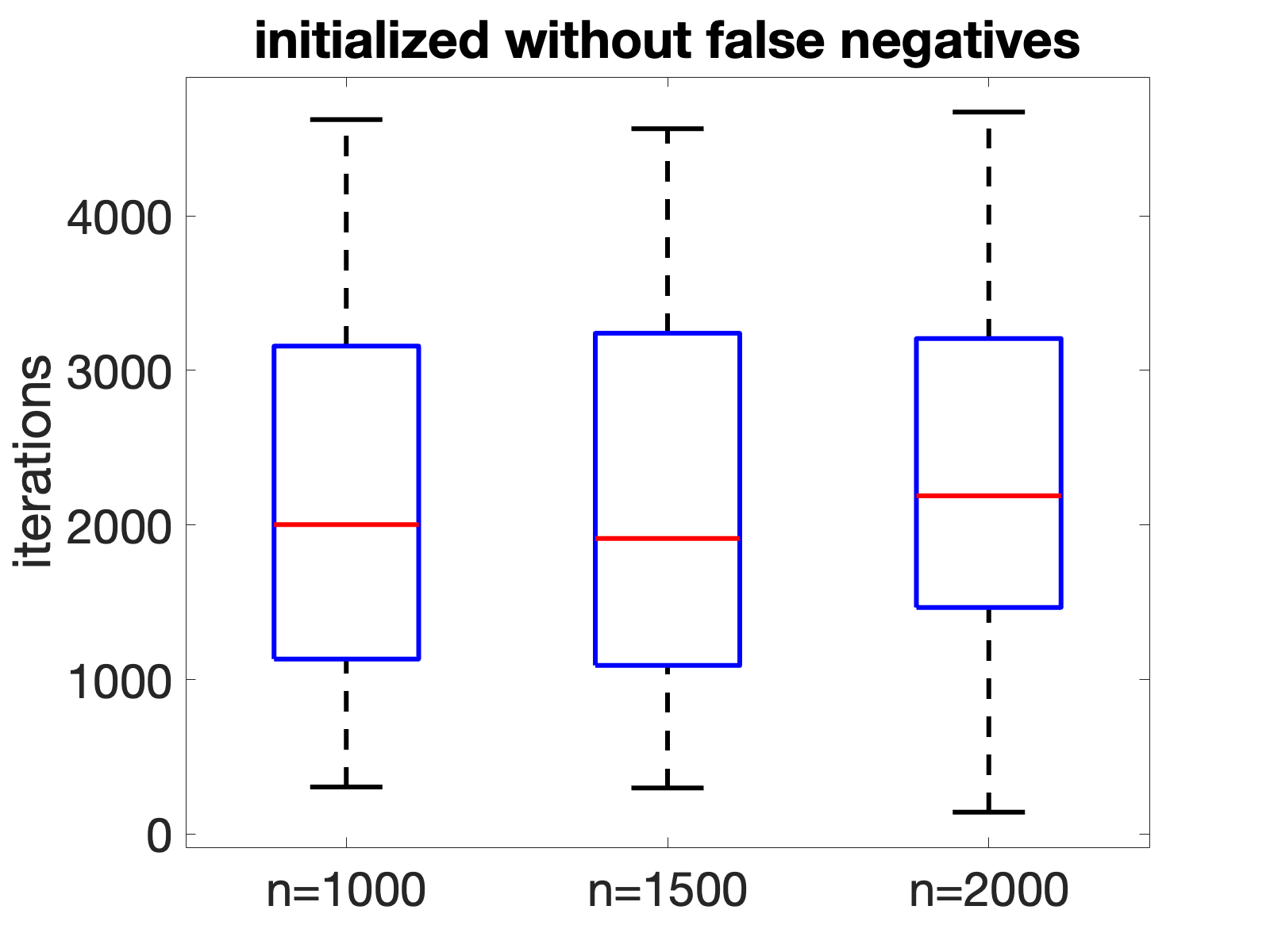}
    \includegraphics[width=.49\textwidth]{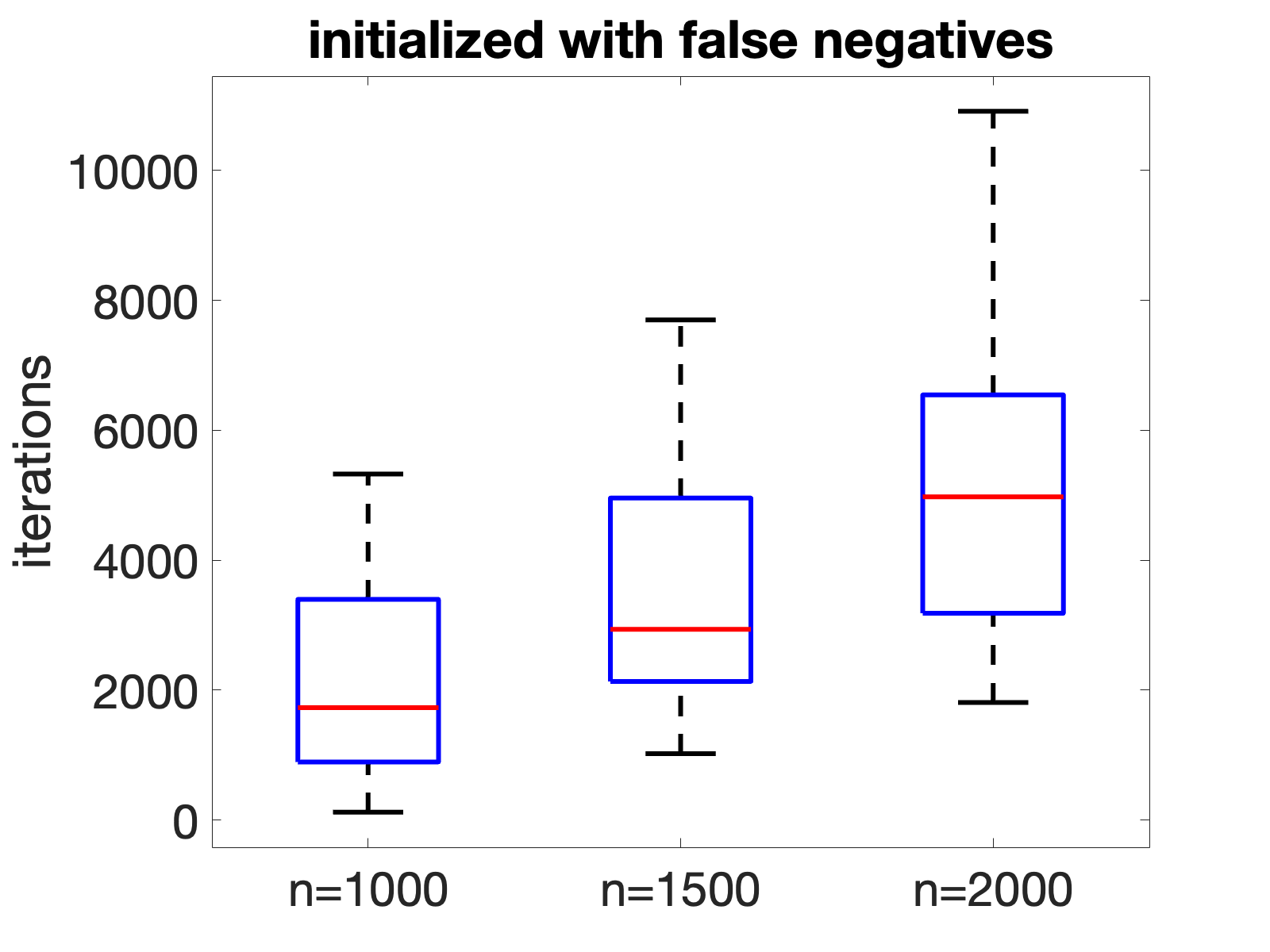}
    \caption{Estimated mixing times of Algorithm \ref{algo:1} under different initialization $\delta^{(0)}$ and for different data size $n$. Left: $\delta^{(0)}$ contains all relevant variables and 10 false-positive. Right: $\delta^{(0)}$ is null model.}
    \label{diff_initialization_logistic}
  \end{figure}


  \subsection{Variable selection consistency of the one-step family}
Since the cost per iteration of Algorithm \ref{algo:1} is polynomial in $(n,p)$, Theorem \ref{thm:post:contr} and Theorem \ref{thm:mix} together show that in problems where the one-step family $\{\check\theta^\delta,\;\delta\in\Delta\}$ is variable selection consistent, the variable selection problem can be provably solved with a computational cost that is polynomial in $(n,p)$ using \textsf{OLAP}. But under what condition is the one-step family $\{\check\theta^\delta,\;\delta\in\Delta\}$  variable selection consistent?  We show in this section that if the initial estimator $\wtilde\theta$ is rate optimal (\textsf{lasso} is known to be rate optimal under mild condition), then the one-step family $\{\check\theta^\delta,\;\delta\in\Delta\}$ inherits the variable selection consistency of the MLE family $\{\hat\theta^\delta,\;\delta\in\Delta\}$. This result mirrors the way in which the one-step estimator inherits the efficiency of the MLE in regular models (\cite{vdv:98}~Section 5.7). The variable selection consistency of the MLE family is a more classical problem that has been considered in the literature. For instance, for a class of sparse GLM models, it is shown to hold in \cite{barber2016laplace}~Theorem 2.2. 
We work under the following basic set of assumptions on the data generating process.
  
  \begin{assumption}\label{H:basic}
    \begin{enumerate}
      \item There exists an absolute constant $b<\infty$ such that 
      \begin{equation}\label{norm:Xj}
        \max_{1\leq i\leq n}\;\max_{1\leq j\leq p}\;|{\bf x}_{ij}| \leq b.
      \end{equation}
      \item There exists an absolute constant $M_0<\infty$ such that 
      \[\left\|\sum_{i=1}^n \left(y_i - \psi'(\pscal{\theta_\star}{{\bf x}_i})\right){\bf x}_i\right\|_\infty \leq M_0 \sqrt{n\log(p)}.\]
      \item The function $\psi$ is strictly convex and there exists an absolute constant $c_3>0$ such that for all $x\in\rset$,
      \[\left|\frac{\rmd^{3}\psi(x)}{\rmd x^3}\right| \leq c_3 \frac{\rmd^{2}\psi(x)}{\rmd x^2}.\]
      \item There exists $0<\underline{\kappa}\leq \bar\kappa$ such that for all $\delta_0\subseteq\delta_\star$, for all $w \in \rset^p$, with $\|w\|_0\leq s_\star$ with $\|w\|_2=1$, it holds
      \[\underline{\kappa} n \leq w^\texttt{T}\left(\sum_{i=1}^n\psi^{''}\left(\pscal{{\bf x}_i}{\tilde\theta^{\delta_0}}\right){\bf x}_i{\bf x}_i^\texttt{T}\right)w \leq \bar\kappa n.\]
      \item There exist constants $c$, and $M$ such that for all $1\leq j\neq k\leq p$,
      \[\sup_{\vartheta\in\rset^p:\;\|\vartheta\|_0\leq s_\star,\;\|\vartheta -\theta_\star\|_2\leq c\sqrt{\frac{s_\star\log(p)}{n}}}\;\;\left|\sum_{i=1}^n \psi^{''}(\pscal{\vartheta}{{\bf x}_i}) {\bf x}_{ij}{\bf x}_{ik}\right|\;\leq M \sqrt{n\log(p)}.\]
    \end{enumerate}
  \end{assumption}
  
  \medskip
  \begin{remark}
    H\ref{H:basic}-(2) holds if $y_i - \psi'(\pscal{\theta_\star}{{\bf x}_i})$ is mean-zero and sub-Gaussian, which holds for many GLM models, including linear, logistic, and Poisson regression. Indeed, if data $\D\eqdef\{(y_i,{\bf x}_i),\;1\leq i\leq n\}$ is a sequence of independent and identically distributed random variables, and the sub-Gaussian norm of $y_i - \psi'(\pscal{\theta_\star}{{\bf x}_i})$ is $\sigma^2$, then by Hoeffding's inequality (\cite{vershynin:18}~Theorem 2.6.2), and a union bound inequality, there exists an absolute constant $c_0$ such that 
    \begin{equation}\label{sub:gauss:ineq}
      \PP\left(\left\|\sum_{i=1}^n \left(y_i - \psi'(\pscal{\theta_\star}{{\bf x}_i})\right){\bf x}_i\right\|_\infty >c_0 \sigma \sqrt{bn\log(p)}\right)\leq \frac{2}{p}.\end{equation}
    
    Assumption H\ref{H:basic}-(3) is a self-concordance assumption that is satisfied by commonly used functions including in linear, logistic and Poisson regression models. For instance for logistic regression, $\psi(x) =\log(1 +e^x)$, and
    \[\left|\frac{\rmd^{3}\psi(x)}{\rmd x^3}\right|=\left|\frac{e^x(e^x-1)}{(1+e^x)^3}\right|\leq \frac{e^x}{(1+e^x)^2} = \frac{\rmd^{2}\psi(x)}{\rmd x^2}.\]
    Hence H\ref{H:basic}-(3) holds with $c_3=1$. For Poisson regression $\psi(x) =e^x$, hence H\ref{H:basic}-(3) also holds with $c_3=1$.
    
H\ref{H:basic}-(4) can be shown to follow from a restricted eigenvalue assumption on the sample covariance matrix $X^\texttt{T} X/n$, where $X\in\rset^{n\times p}$ is the matrix of covariates with $i$-th row given by ${\bf x}_i'$. To see this, note that for most reasonable initial estimators, with high probability it holds $\|\tilde\theta\|_1\leq b_2$ for some constant $b_2$ that can be assumed independent of $(n,p)$. Therefore in that case, using H\ref{H:basic}-(1), $|\pscal{{\bf x}_i}{\tilde\theta^{\delta_0}}|\leq b \|\tilde\theta\|_1 \leq B$ for some constant $B$. Thus, letting $\tau =\inf_{x\in [-B,B]} \psi^{''}(x)>0$, we see that with high probability we have
\[
w^\texttt{T}\left(\sum_{i=1}^n\psi^{''}\left(\pscal{{\bf x}_i}{\tilde\theta^{\delta_0}}\right){\bf x}_i{\bf x}_i^\texttt{T}\right) w =\sum_{i=1}^n\psi^{''}\left(\pscal{{\bf x}_i}{\tilde\theta^{\delta_0}}\right)\pscal{{\bf x}_i}{w}^2 \geq \tau w^\texttt{T} (X^\texttt{T} X)w.
\]
The restricted eigenvalue assumption on the sample covariance matrix $X^\texttt{T} X/n$ is known to hold in setting where the matrix  $X$ can be viewed as a realization of a random matrix with i.i.d rows drawn from a sub-Gaussian distribution. We refer the reader for instance to \cite{Wainwright_2019}~Chapter 9  for more details.

H\ref{H:basic}-(5) imposes a bound on the off-diagonal elements of the Hessian matrices of the log-likelihood. In the linear case, this corresponds to the limited coherence assumption of the design matrix commonly imposed (\cite{candes:plan:09}). 
See also \cite{barber:drton:15} for a similar assumption in a logistic regression setting. 
  \end{remark}
  
  \medskip

  \begin{theorem}\label{lem:post:consis}
    Assume H\ref{H:basic}, with $p\geq 4$. Suppose that the initial estimator $\wtilde\theta$ and the oracle MLE $\hat\theta^{\delta_\star}$ satisfy
    \begin{equation}\label{rate:estimators}
      \max\left(\|\hat\theta^{\delta_\star} - \theta_\star\|_2,\; \|\tilde\theta - \theta_\star\|_2\right) \leq C\sqrt{\frac{s_\star\log(p)}{n}},
    \end{equation}
    for some constant $C$. If the MLE family $\{\hat\theta^\delta,\;\delta\in\Delta\}$ is model selection consistent, then we can find a constant $C'$ such that with a sample size 
    \begin{equation}\label{cond:ss:2}
      n \geq C' s_\star^4 \log(p),
    \end{equation}
    the one-step family $\{\check\theta^\delta,\;\delta\in\Delta\}$ is also model selection consistent.
  \end{theorem}
  \begin{proof}
See Section \ref{sec:proof:post:consis}.
  \end{proof}

\section{Approximate spectral gaps for Markov kernels}\label{sec:mix:mc}
  The proof of Theorem \ref{thm:mix} is based on some general results of independent interest that we present in this section, and which  provide lower bounds on the mixing time of Markov kernels. Specifically, Theorem \ref{thm:mix} is proved by applying Proposition \ref{lb:cp} and Lemma \ref{lem:key:sprofile} below to the Gibbs sampler chain produced by Algorithm \ref{algo:1}.

Let $\pi$ be a density on some measure space $\Xset$  equipped with a sigma-algebra  $\B$ and a reference sigma-finite measure $\rmd x$. We will also write $\pi$ to denote the induced probability measure: $\pi(\rmd x) = \pi(x)\rmd x$.
  For a function $f:\;\Xset\to\rset$, we write $f\in\B$ to say that $f$ is $\B$-measurable. We also define
  \[\|f\|_\infty \eqdef \sup_{x\in\Xset}|f(x)|,\;\;\mbox{ and }\;\; \texttt{osc}(f) \eqdef \sup_{y,z\in\Xset}|f(y) - f(z)|.\]
  We let $L^2(\pi)$ denote the Hilbert space of all real-valued square-integrable (wrt $\pi$) functions on $\Xset$, equipped with the inner product $\pscal{f}{g}\eqdef\int_\Xset f(x)g(x)\pi(\rmd x)$ with associated norm $\|\cdot\|_2$. We will also make use of  the essential supremum (resp. essential infimum) of $f$ with respect to $\pi$ defined as $\texttt{ess-sup}(f) \eqdef\inf\{M\geq 0:\;\pi(\{x\in\Xset:\; |f(x)|>M\})=0\}$ (resp. $\texttt{ess-inf}(f) \eqdef\sup\{M\geq 0:\;\pi(\{x\in\Xset:\; |f(x)|<M\})=0\}$).
  
  If $K$ is a Markov kernel on $\Xset$, and $n\geq 1$ an integer, $K^n$ denotes the $n$-th iterate of $K$, defined recursively as $K^n(x,A) \eqdef \int_\Xset K^{n-1}(x,\rmd z) K(z,A)$, $x\in\Xset$, $A$ measurable. For $f\in\B$, we define  $Kf:\;\Xset\to\rset$ as  $Kf(x) \eqdef \int_\Xset K(x,\rmd z)f(z)$, $x\in\Xset$, whenever the integral is well defined. And if $\mu$ is a probability measure on $\Xset$, then $\mu K$ is the probability on $\Xset$ defined as $\mu K(A) \eqdef\int_\Xset \mu(\rmd z) K(z,A)$, $A\in\B$. 
  The total variation distance  between two probability measures $\mu,\nu$ is defined as
  \[\|\mu-\nu\|_\tv \eqdef 2 \sup_{A \in \B} \left(\mu(A)-\nu(A)\right) = \sup_{f\in\B:\;\|f\|_\infty\leq 1}\left(\int_\Xset f(x)\mu(\rmd x) - \int_\Xset f(x)\nu(\rmd x)\right).\]
  
  If the Markov kernel $K$ has invariant distribution $\pi$, without changing notation we will also view $K$ as the linear operator on $L^2(\pi)$ that transforms $f$ to $Kf$ as defined above. We write $K^\star$ to denote the adjoint of $K$, that is the linear operator on $L^2(\pi)$ such that $\pscal{Kf}{g} = \pscal{f}{K^\star g}$ for all $f,g\in L^2(\pi)$. We say that $K$ is  reversible with respect to $\pi$ ($\pi$-reversible for short) if $K^\star = K$, and we say that $K$ is positive if it is $\pi$-reversible and $\pscal{f}{Kf}\geq 0$ for all $f\in L^2(\pi)$.   For $f\in L^2(\pi)$, we set $\pi(f)\eqdef \int_\Xset f(x)\pi(\rmd x)$, $\textsf{Var}_\pi(f) \eqdef \|f-\pi(f)\|_2^2$,  and 
  \[\e_K(f,f)\eqdef \frac{1}{2}\int_\Xset\int_\Xset (f(y)-f(x))^2\pi(\rmd x) K(x,\rmd y) = \pscal{f}{f} - \pscal{f}{Kf}.\]  
To quantify the rate of convergence of $K^n$ towards $\pi$, we will use the concept of approximate spectral gap  (\cite{atchade:asg}).   For $\zeta\in [0,1)$, the $\zeta$-spectral gap of $K$ is
\begin{equation}\label{eq:zeta:sg}
\lambda_{\zeta}(K) \eqdef \inf\left\{\frac{\e_{K}(f,f)}{\textsf{Var}_{\pi}(f) - \frac{\zeta}{2}},\;\;f\in L^2(\pi)\;\mbox{ s.t. }\; \|f\|_\infty\leq 1\;\mbox{ and }\; \textsf{Var}_{\pi}(f)> \zeta\right\}.\end{equation}
We recover the classical spectral gap (denoted $\lambda(K)$) by taking $\zeta=0$, and replacing the infinity norm by the $L^2$-norm\footnote{In the definition of $\lambda_{\zeta}(K)$ one can replace the infinity norm $\|\cdot\|_\infty$ by any other norm that satisfies $\|Kf\|\leq \|f\|$ for all $f\in L^2(\pi)$. We use the infinity norm here mainly for convenience.}. We note that when $K$ is positive then $0\leq \lambda_\zeta(K) \leq 2$. This quantity can be used to  quantify the convergence rate of $K$ toward a total variation ball of radius  $O(\sqrt{\zeta})$ around $\pi$. Specifically, the following is a slight modification of  Lemma 2.1 of \cite{atchade:asg}. We provide a proof for completely. The reversibility or positivity of $K$ are not needed, but are imposed here for simplicity.

\begin{lemma}\label{lem:key:sprofile}
Suppose that $K$ is $\pi$-reversible and positive, and let $\zeta\in [0,1)$.  Let $\pi_0(\rmd x) = f_0(x)\pi(\rmd x)$, for some bounded function $f_0$.  For all integer $N\geq 1$,  we have
\begin{equation}\label{lem:key:eq:pos}
\|\pi_0 K^N-\pi\|_\tv^2  \leq \max\left(\zeta\|f_0\|_\infty^2,\; \left[1 - \frac{\lambda_\zeta(K)}{2}\right]^N\textsf{Var}_\pi(f_{0})\right).
\end{equation}
\end{lemma}
\begin{proof}
See Section \ref{sec:proof:lem:key}.
\end{proof}

For $N\geq 2\log(\zeta^{-1})/\lambda_\zeta(K)$, we have
\[\left(1 - \frac{\lambda_\zeta(K)}{2}\right)^N\leq \exp\left(-\frac{N\lambda_\zeta(K)}{2}\right) \leq \zeta.\]
Hence it follows from the theorem that for $N\geq 2\log(\zeta^{-1})/\lambda_\zeta(K)$, we have $\|\pi_0 K^N-\pi\|_\tv\leq \|f_0\|_\infty\sqrt{\zeta}$. We note however that the right-hand side of (\ref{lem:key:eq:pos}) does not converge to $0$ as $N\to\infty$: unlike the spectral gap, the approximate spectral gap measures only converge to within $O(\sqrt{\gamma})$ of $\pi$. On the plus side, the method has the advantage that in many problems it scales better with the dimension of the problem than $\lambda(K)$.

\subsection{Conductance and Cheeger's inequality}
A similar concept that pre-dates and has motivated the development of the approximate spectral gap is the $\zeta$-conductance of \cite{lovasz:simonovits93}. For $\zeta\in [0,1/2)$, the $\zeta$-conductance of the Markov kernel $K$ is defined  as 
\[\Phi_\zeta(K)\eqdef \inf\left\{\frac{\int_A\pi(\rmd x)K(x,A^c)}{(\pi(A)-\zeta)(\pi(A^c)-\zeta)},\;A\in\B:\;\zeta<\pi(A) < 1-\zeta\right\}.\]
The special case $\Phi_0(K)$ (that is, $\zeta=0$)  corresponds to the standard conductance (\cite{lawler:sokal:88,sinclair:jerrum89,douc:etal:18}), and we will also denote it by $\Phi(K)$. The conductance $\Phi(K)$ captures how rapidly a Markov chain in stationarity is able to move around the space. $\Phi_\zeta(K)$ measures a similar feature but ignore small sets. Corollary 1.5 of \cite{lovasz:simonovits93} shows that with a warm start, a Markov chain with transition kernel $K$ converges to within $O(\zeta)$ of $\pi$ in at most $2\log(1/\zeta)/\Phi_\zeta(K)^2$ number of iterations. This is analogous to Lemma \ref{lem:key:sprofile}. Here too, the $\zeta$-conductance often depends more favorably on the dimensional of the problem than $\Phi(K)$, however it only measures convergence to within $O(\zeta)$ to $\pi$. In many problems, particularly when dealing with log-concave distributions where a rich set isoperimetric inequalities are available, the conductances are typically easier  to control than the spectral gaps (see for instance \cite{chen:etal:2020} for references and some recent results).

The relationship between the conductance and the spectral gap has been investigated in the literature, and is captured by Cheeger's inequality (\cite{lawler:sokal:88}) that states that
\begin{equation}\label{basic:cheeger}
\frac{\Phi(K)^2}{8} \leq \lambda(K) \leq \Phi(K).\end{equation}
A similar relationship is expected to hold between $\Phi_\zeta(K)$ and $\lambda_\zeta(K)$. We close the gap in the literature by showing that this is indeed the case.

\begin{theorem}\label{lem:cheeger}
For all $\epsilon \in [0,1/2)$. We have
\[ \frac{\left(\Phi_{\frac{\epsilon}{32}}(K)\right)^2}{16}\leq \lambda_{\frac{\epsilon}{2}}(K)\leq \Phi_\epsilon(K).\]
\end{theorem}
\begin{proof}
See Section \ref{sec:proof:lem:cheeger}.
\end{proof}

\subsection{Bounds using canonical paths}
Moving closer to our intended application, we assume now that $\Xset$ is a discrete set. We recall some basic definitions from graph theory. A graph $(\mathcal{V},\e)$ with vertex set $\mathcal{V}$ and edge set $\e$ is a set $\mathcal{V}$ together with a subset $\e$ of $\mathcal{V}\times \mathcal{V}$. We say that $(\mathcal{V},\e)$ is undirected if for all $(x,y)\in\mathcal{V}\times\mathcal{V}$, $(x,y)\in\e$ if and only if $(y,x)\in\e$. We will write an edge as $(e_-,e_+)$, where $e_-,e_+$ denote its two incident nodes. We say that the graph is connected if for all $(x,y)\in\mathcal{V}\times\mathcal{V}$, $x\neq y$, we can find a sequence of edges $(z_0,z_1),\ldots,(z_{\ell-1},z_\ell)$ such that $z_0=x$, $z_\ell=y$, and $(z_{k-1},z_k)\in\e$ for $k=1,\ldots,\ell$. The integer $\ell$ is the length of the path. There may exist many paths linking any two points. For each pair $(x,y)\in \mathcal{V}$, $x\neq y$, we assume given a special path denoted $\gamma_{xy}$ linking $(x,y)$ that we call a canonical path. We impose the additional restriction that an edge can appear only once along a given canonical path. We write $|\gamma_{xy}|$ to denote the length of the canonical path $\gamma_{xy}$, and $\Gamma \eqdef \{\gamma_{xy},\;x, y\in\mathcal{V}\}$ for the set of all canonical paths.  We make the following assumption.

\begin{assumption}
\label{H2}
There exists a subset $\Xset_0\subseteq\Xset$ such that $\pi(\Xset_0)>0$, and a connected undirected graph $(\Xset_0,\e_0)$ with canonical paths $\Gamma \eqdef \{\gamma_{xy},\;x, y\in\Xset_0\}$ such that for all $(x,y)\in\e_0$, $\pi(x) K(x,y)>0$.
\end{assumption}

We define 
\begin{equation}\label{def:m}
\textsf{m}(\Xset_0) \eqdef \max_{e\in\e_0} \;\sum_{\gamma_{xy}:\;\gamma_{xy}\ni e}\;\;\frac{|\gamma_{xy}|\pi(x)\pi(y)}{\pi(e_-)K(e_-,e_+)}.\end{equation}
When $\Xset_0=\Xset$, $\m(\Xset_0)$  is the geometric measure of bottleneck of \cite{sinclair:90} (see also  \cite{diaconis:stroock:91}).  In many cases by carefully choosing $\Xset_0$, $\m(\Xset_0)$ scales better than $\m(\Xset)$. The next result is analogous to Proposition 1 of \cite{diaconis:stroock:91}. 

\begin{proposition}\label{lb:cp}
Assume H\ref{H2}. Given $\epsilon\in[0,1/2)$, if $\pi(\Xset_0)\geq 1-\epsilon/8$, then 
$\lambda_\epsilon(K)\geq \mathsf{m}(\Xset_0)^{-1}$.
\end{proposition}
\begin{proof}
See Section \ref{sec:proof:proplbcp}.
\end{proof}

\begin{remark}
Theorem \ref{thm:mix} is proved by applying Proposition \ref{lb:cp} and Lemma \ref{lem:key:sprofile} to the Gibbs sampler chain produced by Algorithm \ref{algo:1}.
\end{remark}

\bigskip

\section{Numerical illustration}\label{sec:num}

We investigate several aspects of Algorithm \ref{algo:1} in a simulation setting using logistic and Poisson regression. Here is the simulation setup. We generate $X\in\rset^{n\times p}$ with independent rows drawn from $\textbf{N}_p(0,\Sigma)$, where $\Sigma_{kj} = \varrho^{|j-k|}$, where $\varrho\in\{0,0.9\}$ referred to as ``low correlation" and as ``highly correlation", respectively. We consider various combinations of $n,p$.

As a true $\theta_\star\in\rset^p$, we use a sparse $\theta_{\star}$ with first 10 components uniformly drawn from $(-3,-2) \cup (2,3)$. For logistic regression, we draw the response as $Y_i \sim \textbf{Ber}(p_i)$, with $p_i = \left(1+ \exp{-\pscal{{\bf x}_i}{\theta_\star}}\right)^{-1}$, where ${\bf x}_i$ denotes the $i$-th row of $X$. For Poisson regression we draw the response as $Y_i \sim \textbf{Poi}(\lambda_i)$, where $\lambda_i = e^{\pscal{{\bf x}_i}{\theta_\star}}$.

Throughout we set the prior parameter $\mathsf{u} = 0.8$, and in Algorithm \ref{algo:1}, we set $J=100$, and unless stated otherwise, we take $\delta^{(0)}$ as the support of \textbf{lasso}.

\subsection{Illustration using logistic regression}

\subsubsection{Effect of the initial estimator}

\begin{figure}[H]
\includegraphics[width=.49\textwidth]{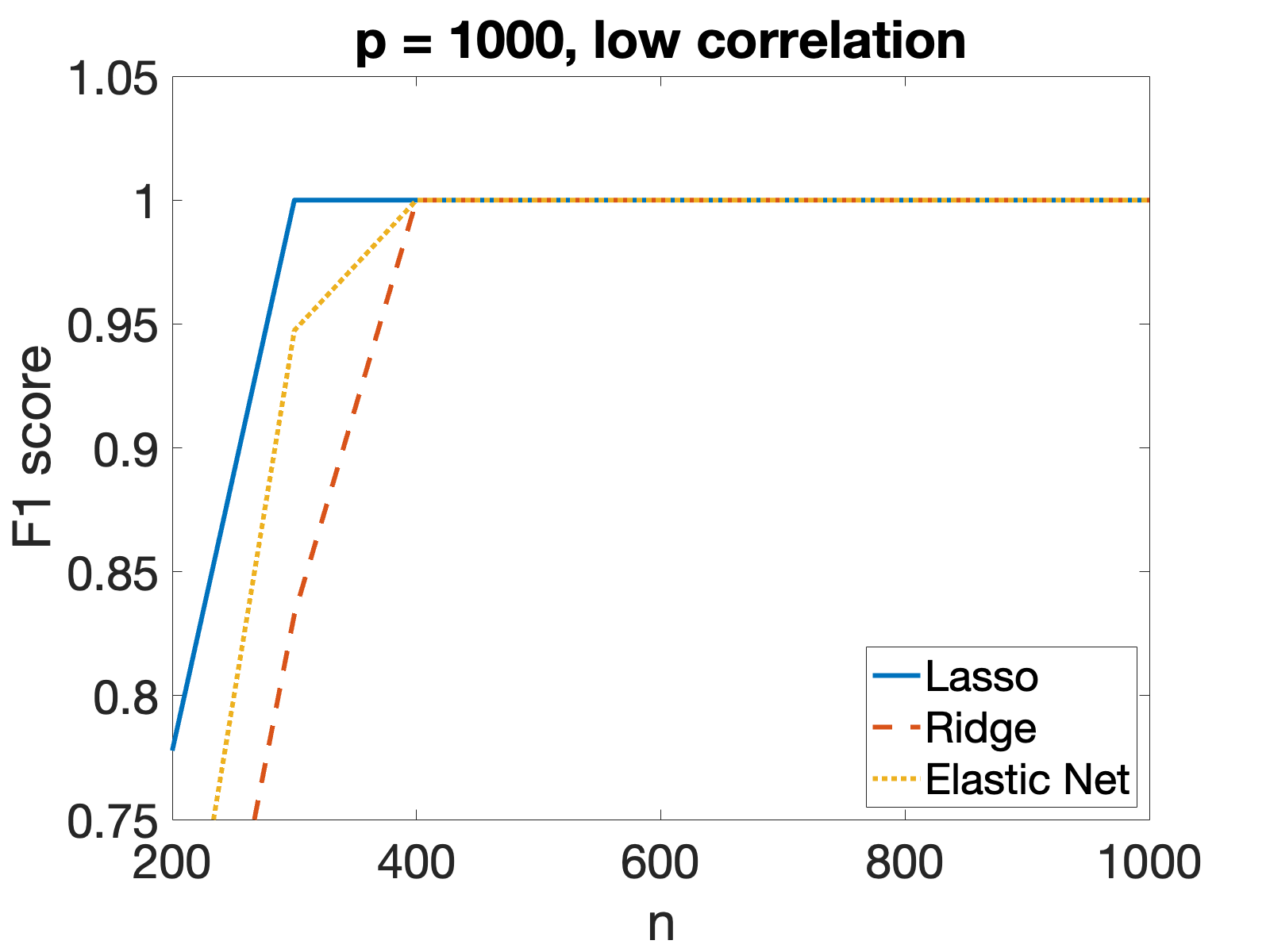}
\includegraphics[width=.49\textwidth]{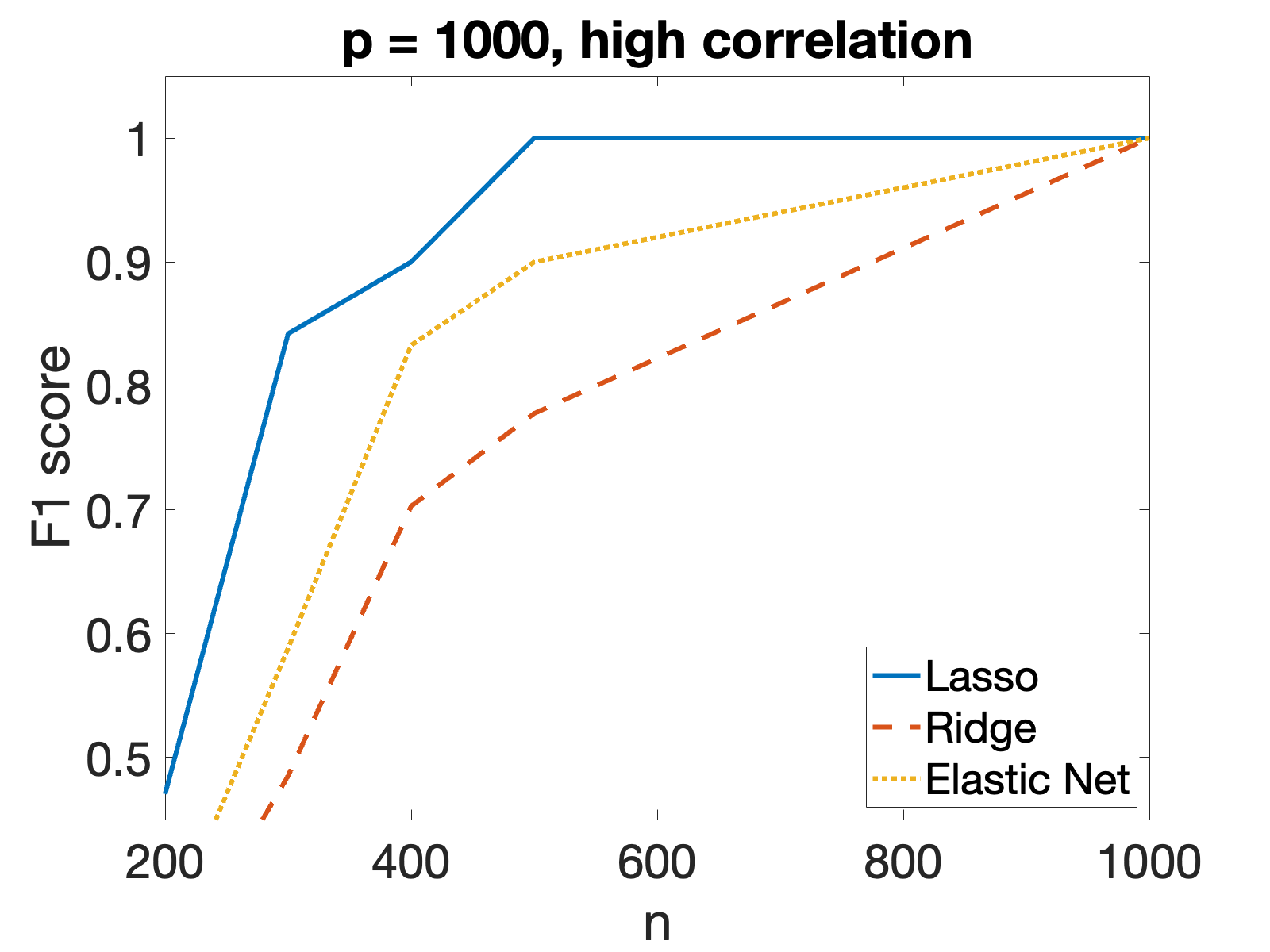}
\vspace{-0.5cm}
\caption{Comparison of F1-scores for different initial estimators for logistic models.}
\label{onestep_diffinit_p1000}
\end{figure}

We first investigate the effect of the initial estimator $\tilde\theta$ on OLAP. We compare three different initializations: \textbf{lasso}, \textbf{ridge}, and \textbf{elastic net}. We utilize R package \textsf{glmnet} to calculate these initial estimators. As a measure of performance we compute the F1-score (harmonic mean of sensitivity and precision) along the MCMC chain.

We focus on the case when $p = 1000$ and we increase the sample size $n$ from $200$ to $1000$, under both the low and high correlation settings. Under each simulation set up, we run Algorithm \ref{algo:1} 50 times, and we report the median of F1-scores of the 50 chains after mixing.

As expected we observe from Figure \ref{onestep_diffinit_p1000} that the F1 score is an increasing function of the sample size, and in all scenarios, the \textbf{lasso} initialization produces  the highest F1-scores. We also observe that given enough sample size, all three initializations perform well, which is consistent with our theoretical findings. Based on these results, we focus on the \textbf{lasso} initialization for the remaining experiments.

\subsubsection{Mixing time comparisons}
We compare empirically the mixing time of Algorithm \ref{algo:1} with the mixing time of the exact method that employs the data-augmentation strategy proposed in \cite{AB:19}. The  method consists in sampling jointly $(\delta,\theta)$ from the distribution 
\begin{equation}\label{post:da}
\Pi(\delta,\theta\vert \D) \propto \left(\frac{1}{p^\mathsf{u}}\sqrt{\frac{1}{\rho_0}}\right)^{\|\delta\|_0} \;\exp\left(-\frac{1}{2}\|\theta_\delta\|_2^2 -\frac{\rho_0}{2}\|\theta-\theta_\delta\|_2^2 + \ell(\theta_\delta;\D)\right),\end{equation}
for some hyper-parameter $\rho_0$. One can sample from this distribution by Metropolis-within-Gibbs, alternating between a Gibbs update on $\delta$ given $\theta$, and a Metropolis-Hastings update on $\theta$ given $\delta$ (here we use MaLA).  We note that the marginal distribution of $\delta$ under (\ref{post:da}) is precisely (\ref{post:Pi}). So the method is exact, and we refer to it as the \textsf{Exact} method. For both algorithms we estimate their mixing times using the $L$-lag coupling method  proposed by \cite{biswas2019estimating}. The $L$-lag coupling method  consists in running coupled version of an MCMC algorithm until a coupling event. The mixing time of the algorithm can then be related to a moment of the coupling time. We refer the reader to \cite{biswas2019estimating} and \cite{jacob2019unbiased}.  In our implementation of the $L$-lag coupling method we couple the chains using a straightforward maximal coupling of the Bernoulli draws in Algorithm \ref{algo:1}. For coupling the Metropolis-within-Gibbs sampling for (\ref{post:da}) we follow  the algorithm in \cite{atchade2021fast}.

\begin{figure}[h]
\includegraphics[width=.49\textwidth]{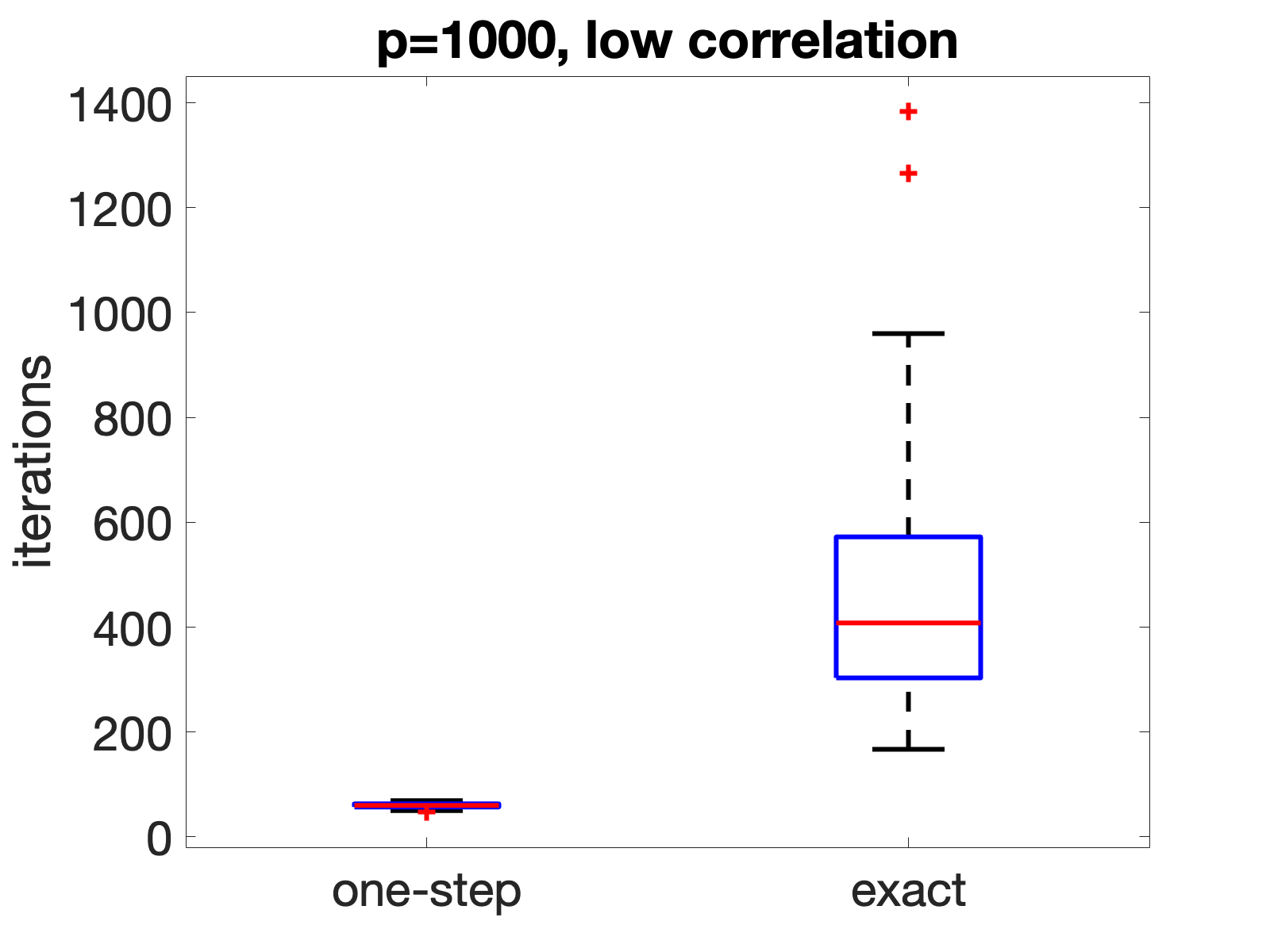}
\includegraphics[width=.49\textwidth]{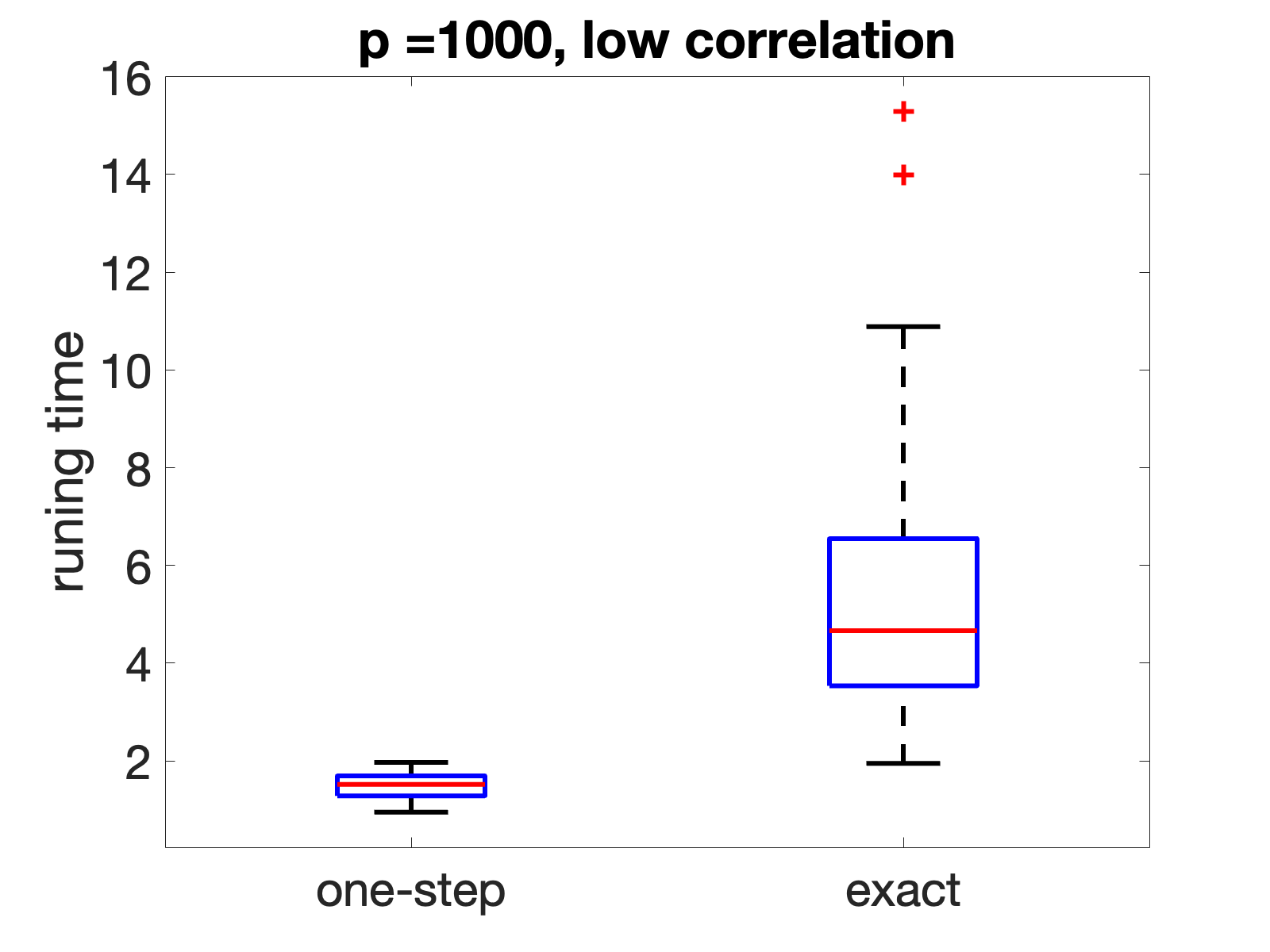}
\vspace{-0.5cm}
\caption{Iterations and mixing time (in seconds) of logistic regression, low correlation.}
\label{mixing_logistic_OLAP_rho0}
\end{figure}

\begin{figure}[h]
\includegraphics[width=.49\textwidth]{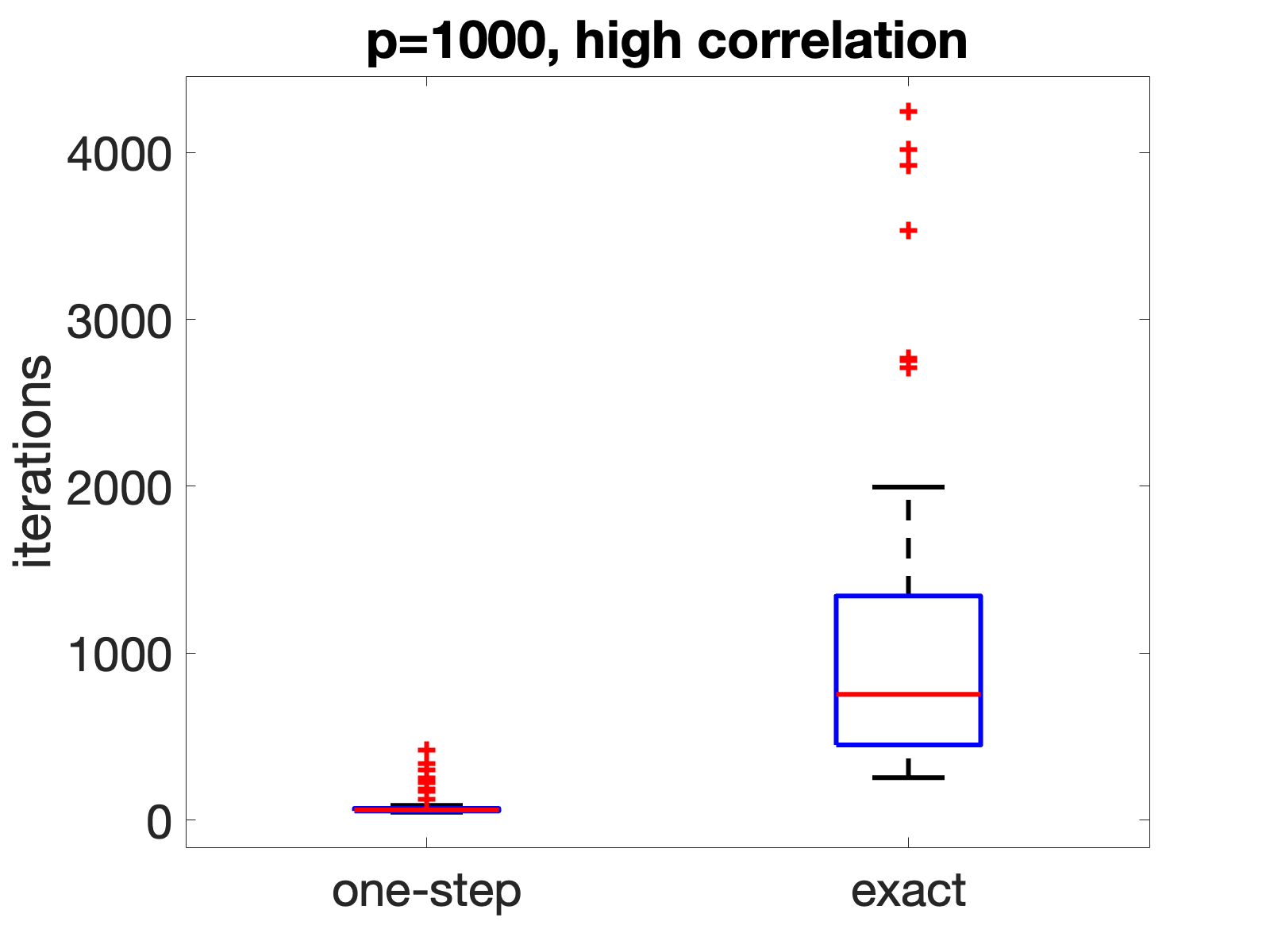}
\includegraphics[width=.49\textwidth]{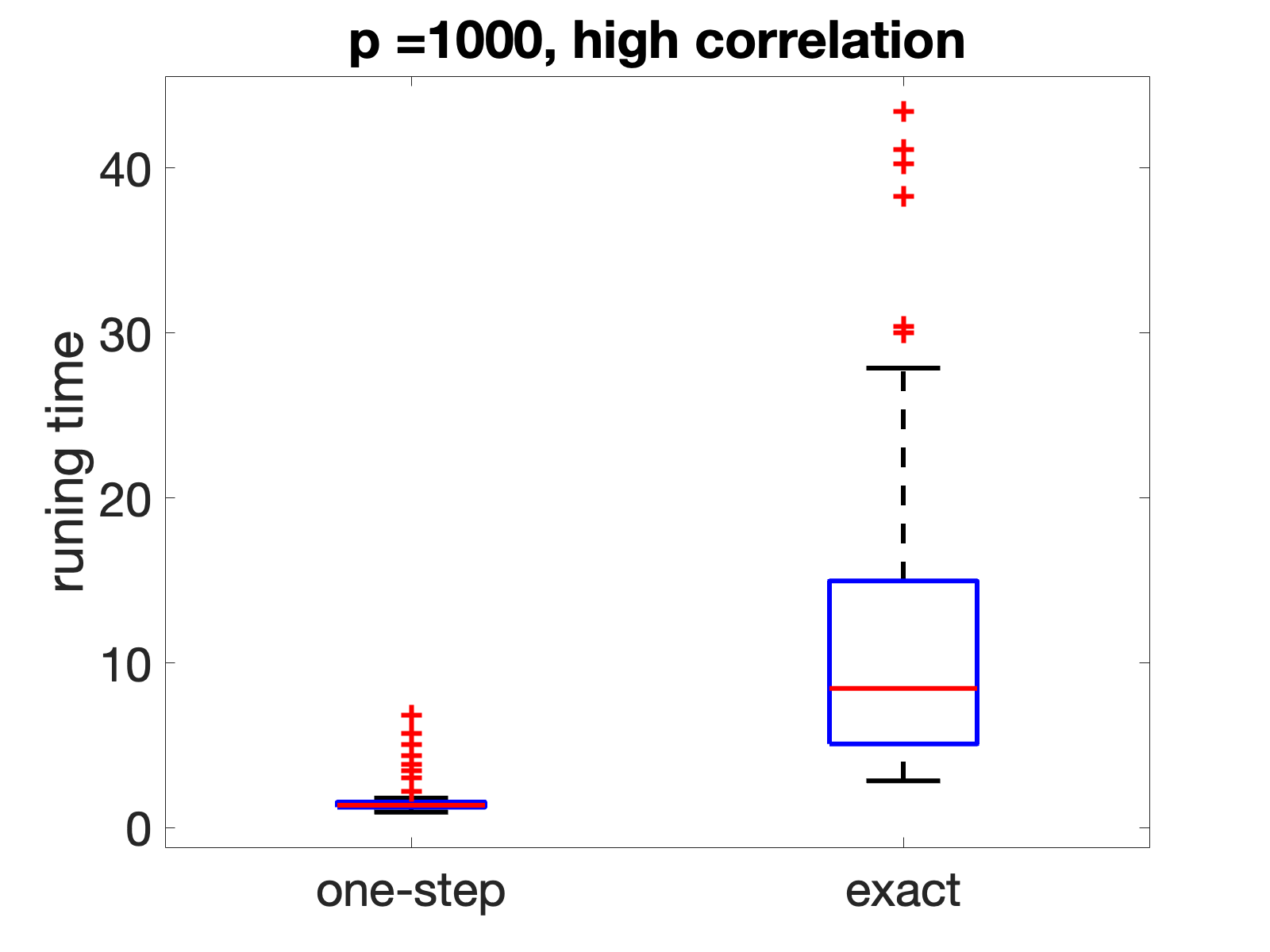}
\vspace{-0.5cm}
\caption{Iterations and mixing time (in seconds) of logistic regression, high correlation.}
\label{mixing_logistic_OLAP_rho09}
\end{figure}

For the comparison we set $p=1000$ and $n =1000$, and we consider both the low and high correlation cases. In each setting, we generate 50 datasets. Each dataset defines a posterior distribution against which the algorithms under consideration have a mixing time. We estimate these mixing times by running 30 coupled chains. We could observe from the boxplots on Figures  \ref{mixing_logistic_OLAP_rho0} and \ref{mixing_logistic_OLAP_rho09} that Algorithm \ref{algo:1} mixes in a much smaller number of iterations and has a shorter running time compared to the Metropolis-within-Gibbs sampler for the \textsf{Exact} method. Specifically, in the low correlation case, \textsf{OLAP} sampler has a median burn-in of $60$ iterations, and a median burn-in running time of $1.5$ seconds, and our sampler for the \textsf{Exact} method has a median burn-in of $408$ iterations, and a median burn-in running time of $4.7$ seconds. In the high correlation case, \textsf{OLAP} sampler has a median burn-in  of $60.5$ iterations, and a median burn-in running time of $1.4$ seconds, and the  sampler for the \textsf{Exact} method has a median burn-in of $753$ iterations, and a median burn-in running time of $8.6$ seconds.

\subsubsection{Statistical performance and comparison with other methods}


\begin{table}[htbp]
\centering
\caption{F1-score for logistic regression. $p = 1000$, low correlation}
\scalebox{0.65}{
\begin{tabular}{|r|r|r|r|r|r|r|r|r|r|r|r|r|}
\toprule
\hline
& \multicolumn{2}{c|}{\textbf{S-Gibbs}} & \multicolumn{2}{c|}{\textbf{SparseVB}} &\multicolumn{2}{c|}{\textbf{One-step Lasso}} &\multicolumn{2}{c|}{\textbf{Exact}} & \multicolumn{2}{c|}{\textbf{Lasso}} \\
\hline
\multicolumn{1}{|l|}{\textbf{n}} & \multicolumn{1}{c|}{\textbf{Median}} & \multicolumn{1}{c|}{\textbf{Std. Error}} & \multicolumn{1}{c|}{\textbf{Median}} & \multicolumn{1}{c|}{\textbf{Std. Error}} & \multicolumn{1}{c|}{\textbf{Median}} & \multicolumn{1}{c|}{\textbf{Std. Error}} & \multicolumn{1}{c|}{\textbf{Median}} & \multicolumn{1}{c|}{\textbf{Std. Error}} & \multicolumn{1}{c|}{\textbf{Median}} & \multicolumn{1}{c|}{\textbf{Std. Error}} \\
\hline
\textbf{200} & 0.778 & 0.142 & 0.889 & 0.092 & 0.778 & 0.093 & 0.750 & 0.220 & 0.300 & 0.080 \\
\hline
\textbf{300} & 0.909 & 0.067 & 1.000 & 0.009 & 1.000 & 0.052 & 1.000 & 0.035 & 0.294 & 0.083 \\
\hline
\textbf{400} & 0.909 & 0.061 & 1.000 & 0.025 & 1.000 & 0 & 1.000 & 0 & 0.280 & 0.098 \\
\hline
\textbf{500} & 0.909 & 0.065 & 1.000 & 0.025 & 1.000 & 0 & 1.000 & 0 & 0.288 & 0.097 \\
\hline
\textbf{1000} & 0.952 & 0.058 & 1.000 & 0.026 & 1.000 & 0 & 1.000 & 0 & 0.339 & 0.122 \\
\bottomrule
\end{tabular}%
}
\label{F1_compare}
\end{table}%

Here we perform a comparison between \textsf{OLAP}, the Exact method, Skinny-Gibbs (\cite{narisetty2018skinny}), SparseVB (\cite{ray2020spike}), and the standard \textbf{lasso} in terms of statistical recovery of the true support $\delta_\star$. For the simulation we use the \textsf{R} packages \textsf{skinnybasad} (for Skinny-Gibbs) and \textsf{svb} (for SparseVB) provided by the authors and we use the popular  \textsf{glmnet} for \textbf{lasso}. We focus on the scenarios with $p = 1000$, and vary the sample sizes and correlations. We report both the medians and standard errors of the F1-scores over 50 chains. The results are presented in Table \ref{F1_compare} and \ref{F1_compare_rho09}.

We could observe from Table \ref{F1_compare} that when $\varrho = 0$,  \textsf{OLAP}, the Exact method, and SparseVB have roughly the same F1 score, and are all better than S-Gibbs and LASSO. While, with high correlation data, we could observe from Table \ref{F1_compare_rho09} that \textsf{OLAP} matches the exact method, and both are better than SparseVB and S-Gibbs.

\begin{table}[htbp]
\centering
\caption{F1-score for logistic regression. $p = 1000$, high correlation}  \scalebox{0.65}{
\begin{tabular}{|r|r|r|r|r|r|r|r|r|r|r|r|r|}
\toprule
\hline
& \multicolumn{2}{c|}{\textbf{S-Gibbs}} & \multicolumn{2}{c|}{\textbf{SparseVB}} &\multicolumn{2}{c|}{\textbf{One-step Lasso}} &\multicolumn{2}{c|}{\textbf{Exact}} & \multicolumn{2}{c|}{\textbf{Lasso}} \\
\hline
\multicolumn{1}{|l|}{\textbf{n}} & \multicolumn{1}{c|}{\textbf{Median}} & \multicolumn{1}{c|}{\textbf{Std. Error}} & \multicolumn{1}{c|}{\textbf{Median}} & \multicolumn{1}{c|}{\textbf{Std. Error}} & \multicolumn{1}{c|}{\textbf{Median}} & \multicolumn{1}{c|}{\textbf{Std. Error}} & \multicolumn{1}{c|}{\textbf{Median}} & \multicolumn{1}{c|}{\textbf{Std. Error}} & \multicolumn{1}{c|}{\textbf{Median}} & \multicolumn{1}{c|}{\textbf{Std. Error}} \\
\hline
\textbf{200} & 0.358 & 0.169 & 0.500 & 0.167 & 0.471 & 0.175 & 0.572 & 0.184 & 0.229 & 0.064 \\
\hline
\textbf{300} & 0.476 & 0.181 & 0.718 & 0.165 & 0.842 & 0.164 & 0.842 & 0.136 & 0.265 & 0.070 \\
\hline
\textbf{400} & 0.586 & 0.162 & 0.800 & 0.143 & 0.900 & 0.096 & 0.900 & 0.105 & 0.281 & 0.076 \\
\hline
\textbf{500} & 0.652 & 0.172 & 0.842 & 0.105 & 1.000 & 0.070 & 0.947 & 0.071 & 0.314 & 0.057 \\
\hline
\textbf{1000} & 0.714 & 0.142 & 1.000 & 0.082 & 1.000 & 0 & 1.000 & 0.037 & 0.310 & 0.052 \\
\bottomrule
\end{tabular}%
}

\label{F1_compare_rho09}
\end{table}%

\subsection{Numerical illustration using Poisson regression}\label{sec:num_poi}
Similarly, we investigate several aspects of Algorithm \ref{algo:1} with a simulated Poisson regression example. The setting is very similar to the logistic experiments, except that, we draw the response as $y_i \sim \textbf{Poi}(\lambda_i)$, with $\lambda_i = \exp(-\pscal{{\bf x}_i}{\theta_\star})$, where ${\bf x}_i$ denotes the $i$-th row of $X$.
\subsubsection{Effect of the initial estimator}
Using a similar experiment as above, we compare the effect of the initial estimator on \textsf{OLAP}. The results is reported on Figure \ref{onestep_diffinit_p1000_Poisson}. We first note that the Poisson regression requires comparatively larger sample size for the convergence of \textsf{OLAP}. And the behavior of the three estimators are much more similar. 



\begin{figure}[H]
\includegraphics[width=.49\textwidth]{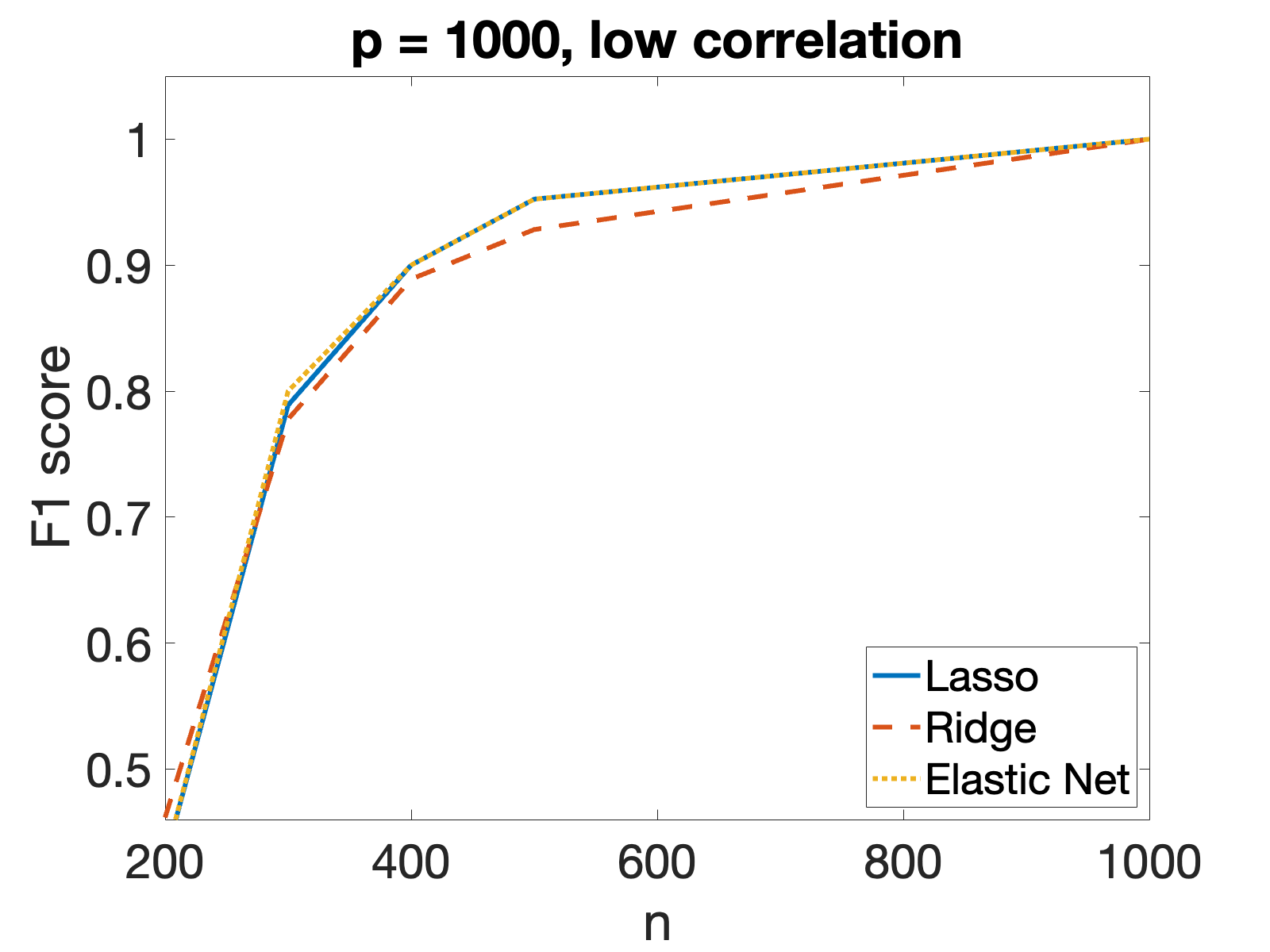}
\includegraphics[width=.49\textwidth]{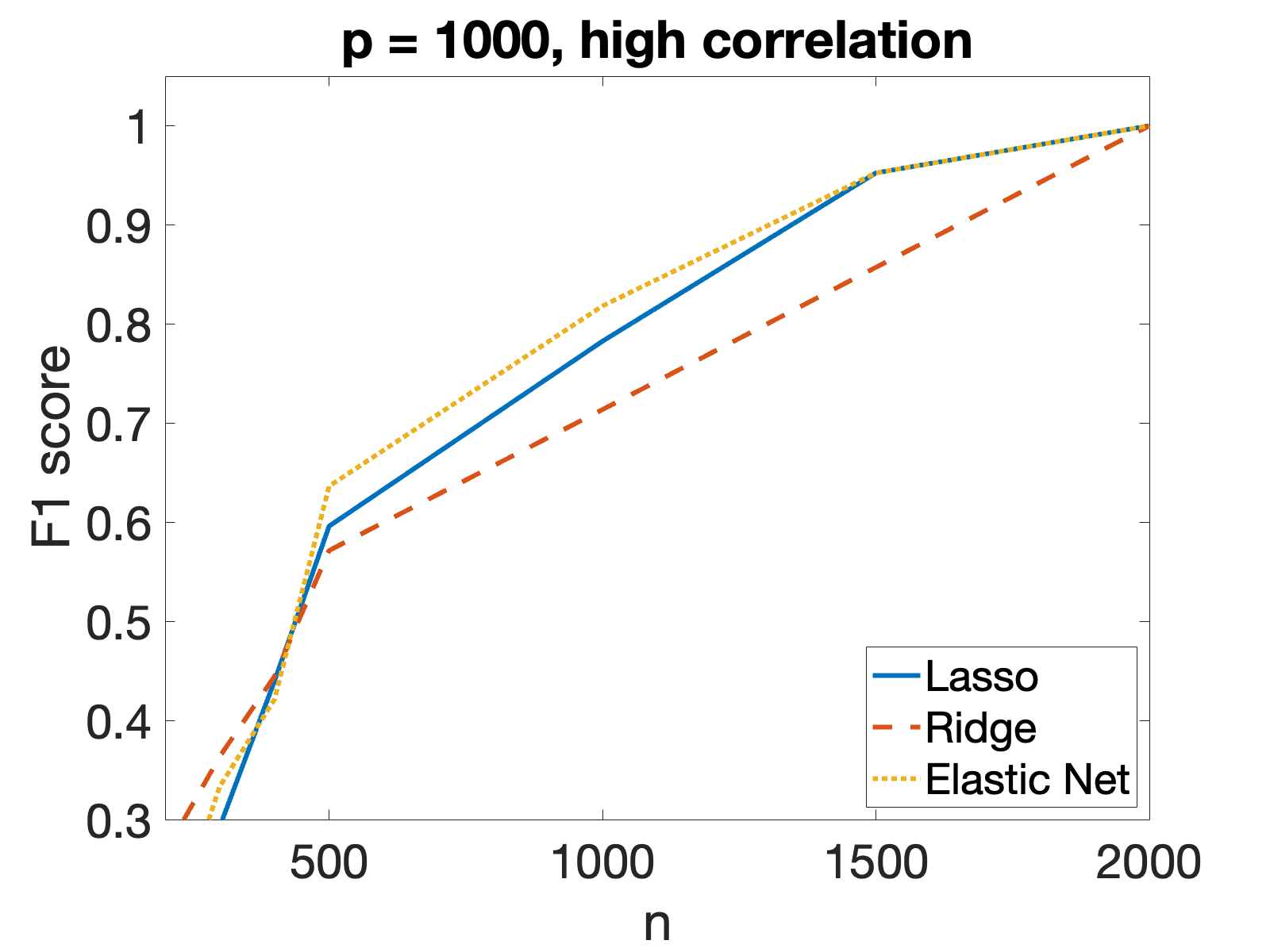}
\caption{Comparison of F1-scores for different initial estimators for Poisson models, low correlation (left) and high-correlation}
\label{onestep_diffinit_p1000_Poisson}
\end{figure}

\subsubsection{Mixing time}
We use the same $L$-lag coupling methodology to compare the mixing time of the \textsf{OLAP} Gibbs sampler and the Metropolis-within-Gibbs sampler for the \textsf{Exact} method. Here we set $n = p=1000$ in the low correlation case, and $n=2000$, $p=1000$ in the high correlation case.

\begin{figure}[h]
\includegraphics[width=.49\textwidth]{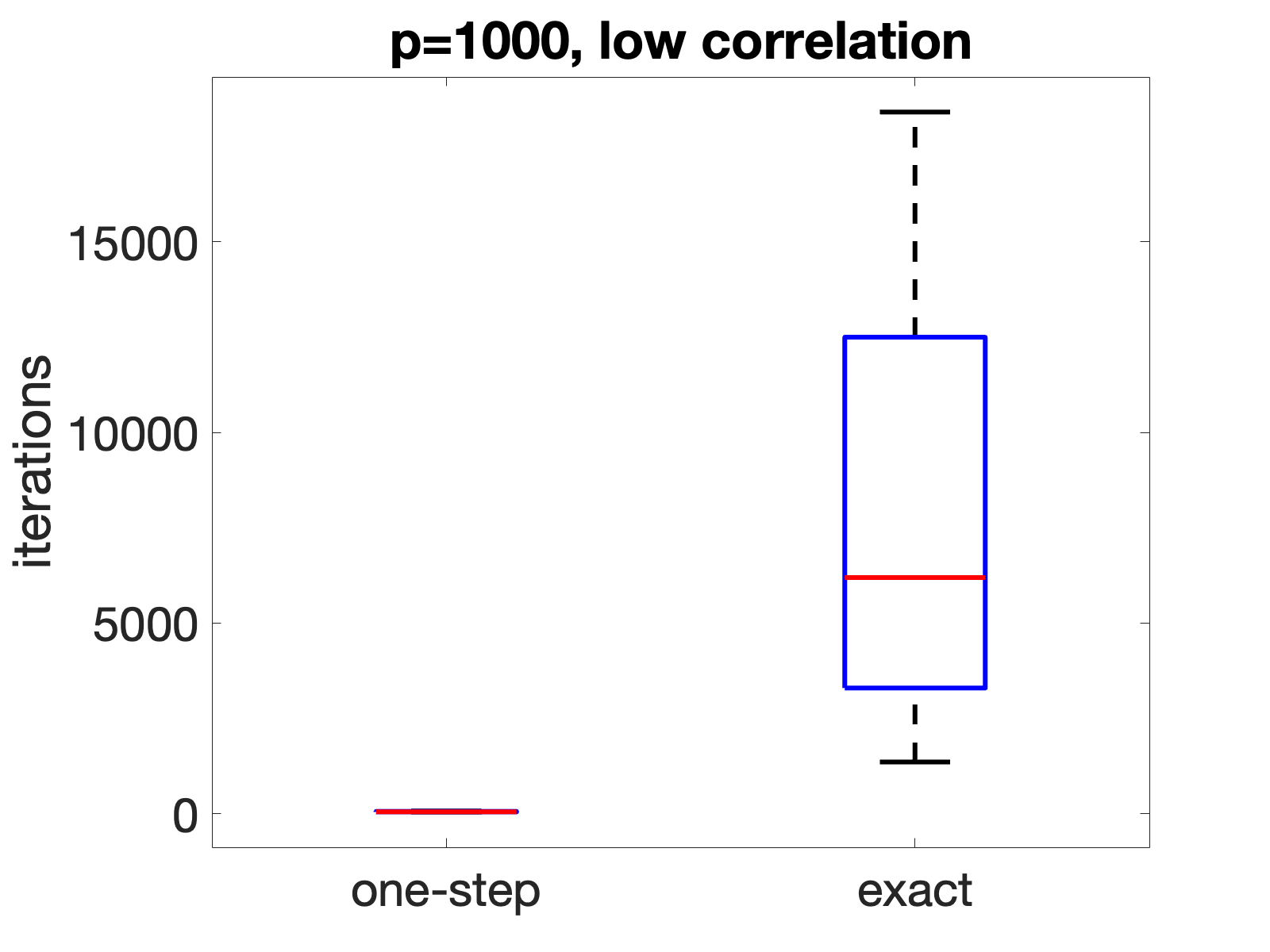}
\includegraphics[width=.49\textwidth]{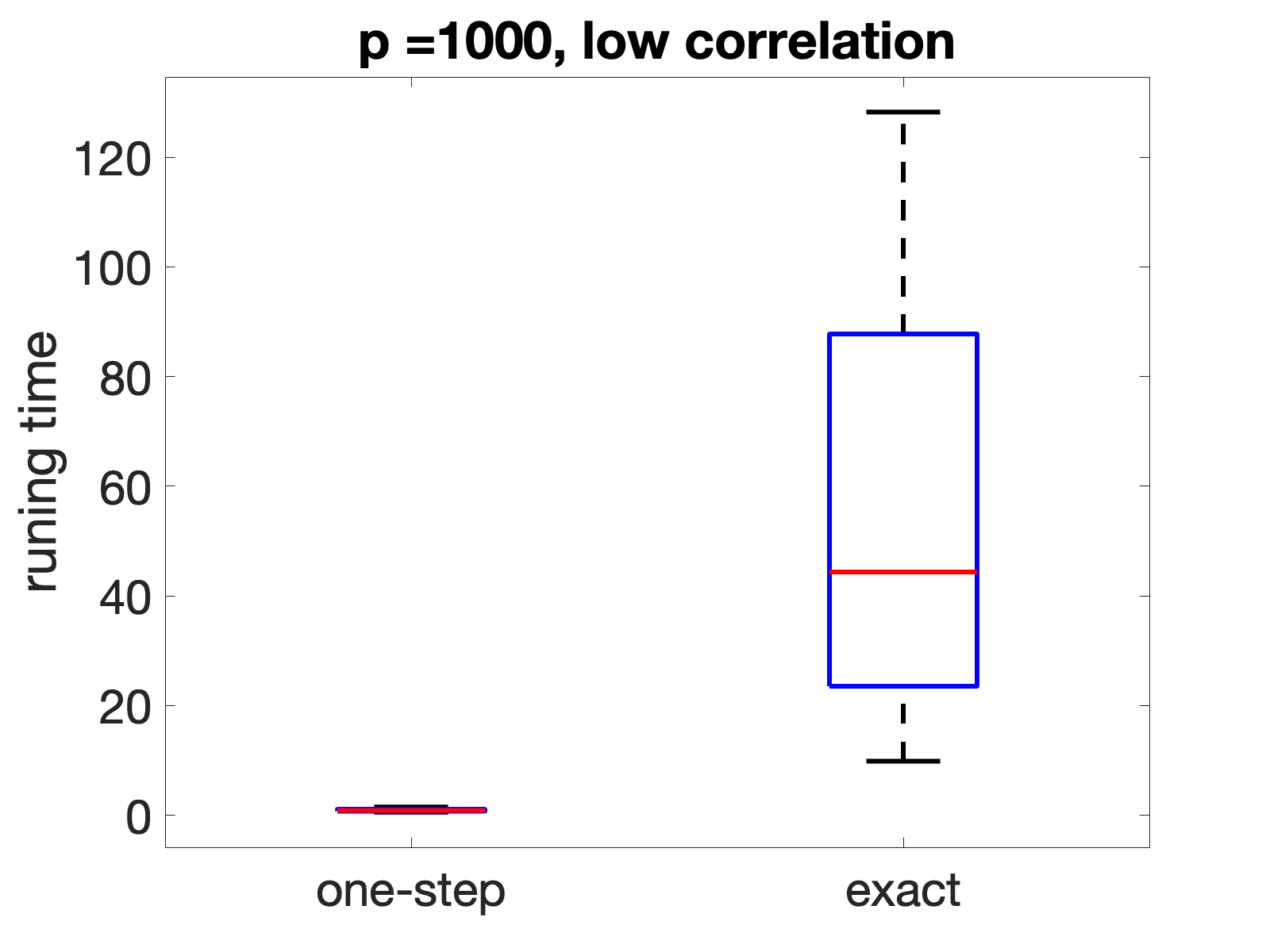}
\caption{Iterations and mixing time (in seconds) of Poisson regression, low correlation.}
\label{mixing_Poisson_OLAP_rho0}
\end{figure}

\begin{figure}[h]
\includegraphics[width=.49\textwidth]{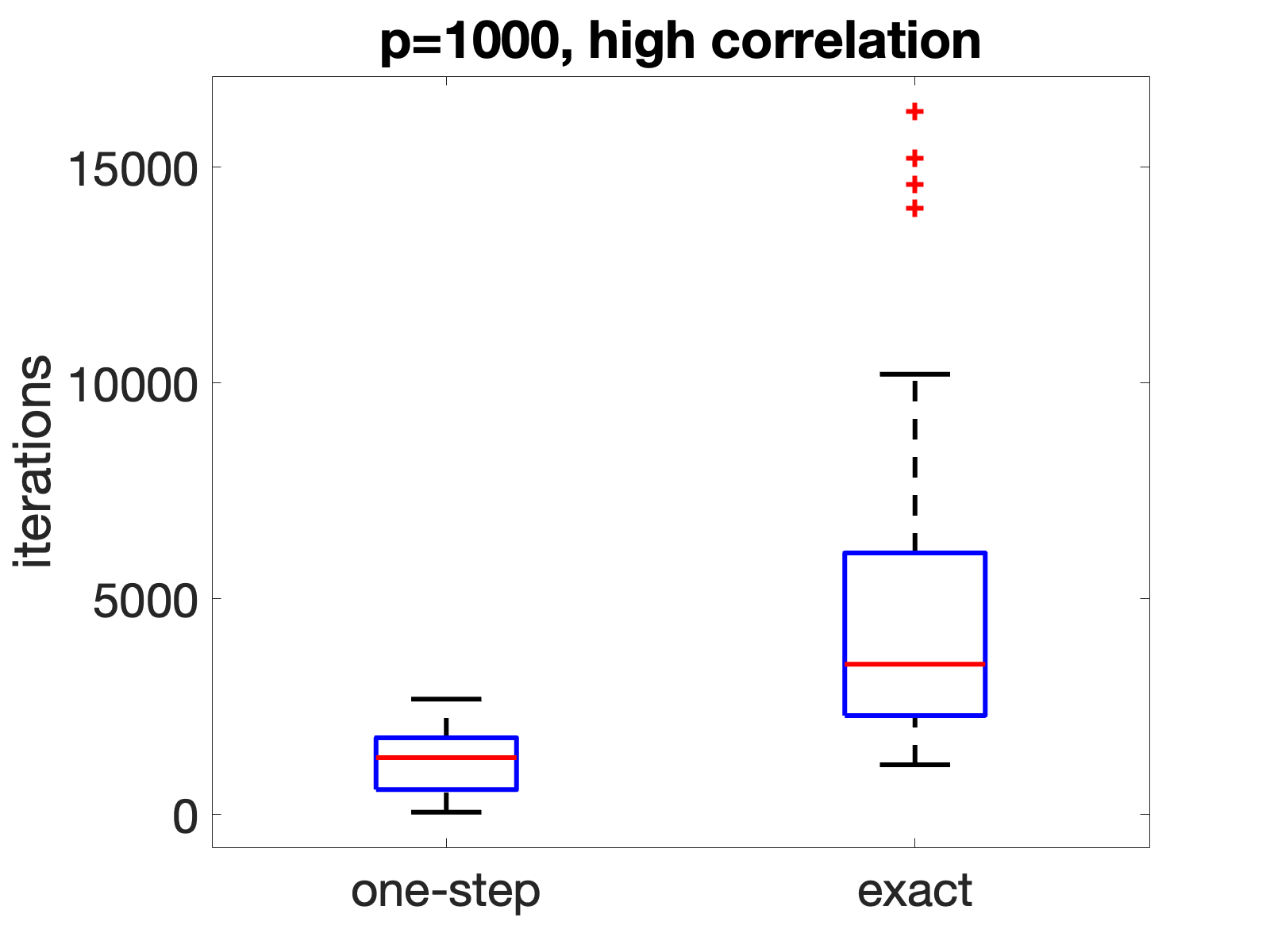}
\includegraphics[width=.49\textwidth]{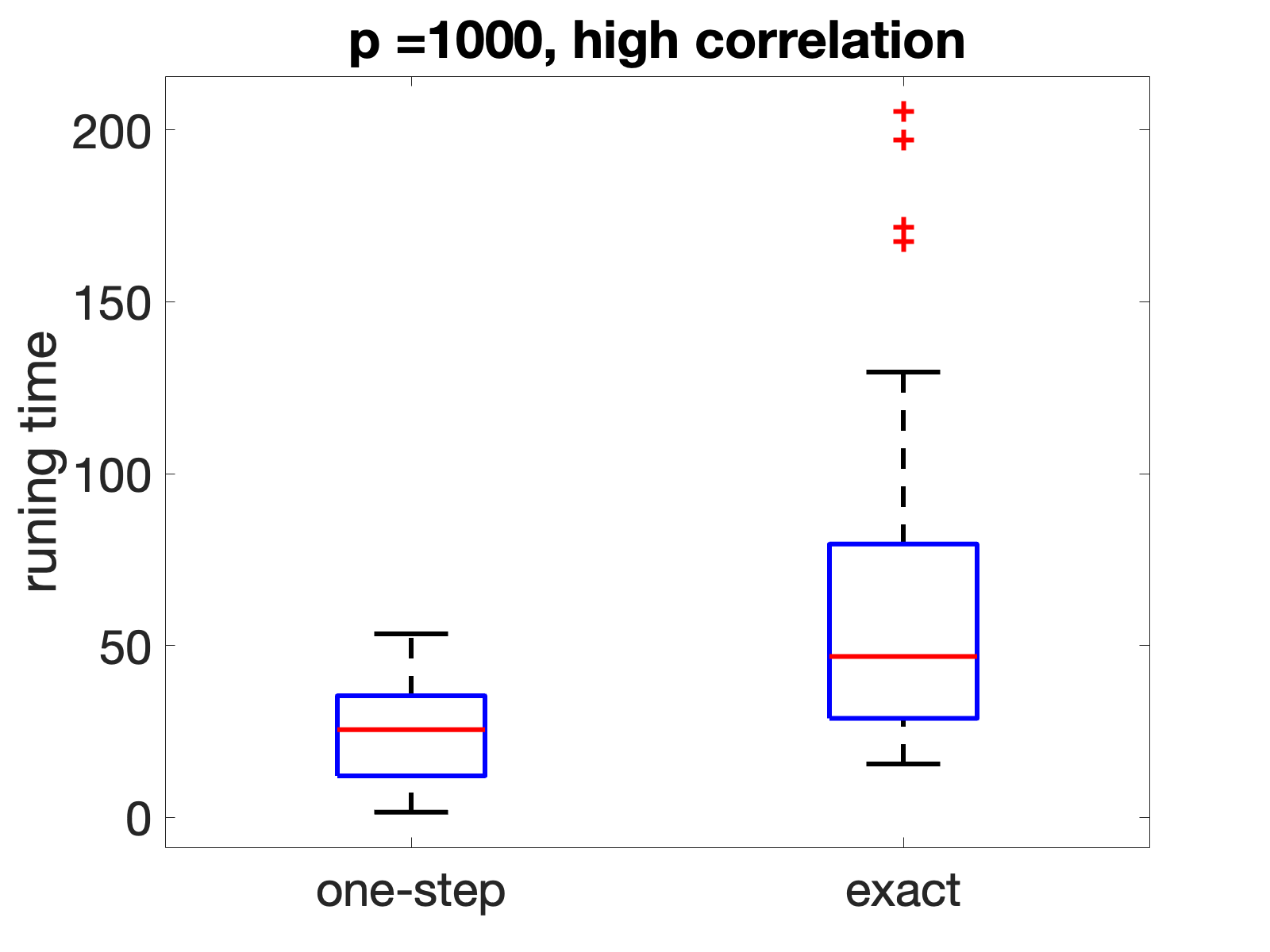}
\caption{Iterations and mixing time (in seconds) of Poisson regression, high correlation.}
\label{mixing_Poisson_OLAP_rho09}
\end{figure}

Again, We could observe from the Boxplot \ref{mixing_Poisson_OLAP_rho09} and \ref{mixing_Poisson_OLAP_rho0} that \textsf{OLAP} sampler takes fewer iterations and less running time than \textsf{Exact} sampler to mix, as in the low correlation case, \textsf{OLAP} sampler has a median burn-in iterations before convergence of $54.5$, and a median burn-in running time of $0.9$ seconds, and \textsf{EXACT} sampler has a median burn-in iterations before convergence of $6186$, and a median burn-in running time of $44.3$ seconds; while in the high correlation case, \textsf{OLAP} sampler has a median burn-in iterations before convergence of $1319.5$, and a median burn-in running time of $25.5$ seconds, and \textsf{EXACT} sampler has a median burn-in iterations before convergence of $3480$, and a median burn-in running time of $46.8$ seconds.


\begin{table}[htbp]
\centering
\caption{F1-score for Poisson regression. $p = 1000$,  low correlation}
\scalebox{0.75}{
\begin{tabular}{|r|r|r|r|r|r|r|r|r|r|r|r|r|}
\toprule
\hline
\hline
& \multicolumn{2}{c|}{\textbf{One-step Lasso}} & \multicolumn{2}{c|}{\textbf{Exact}} & \multicolumn{2}{c|}{\textbf{Lasso}} \\
\hline
\multicolumn{1}{|l|}{\textbf{n}} & \multicolumn{1}{c|}{\textbf{Median}} & \multicolumn{1}{c|}{\textbf{Std. Error}} & \multicolumn{1}{c|}{\textbf{Median}} & \multicolumn{1}{c|}{\textbf{Std. Error}} & \multicolumn{1}{c|}{\textbf{Median}} & \multicolumn{1}{c|}{\textbf{Std. Error}} \\
\hline
\textbf{200} & 0.429 & 0.135 & 0.533 & 0.170 & 0.300 & 0.080 \\
\hline
\textbf{300} & 0.789 & 0.133 & 0.833 & 0.131 & 0.294 & 0.083 \\
\hline
\textbf{400} & 0.900 & 0.107 & 0.947 & 0.197 & 0.280 & 0.098 \\
\hline
\textbf{500} & 0.952 & 0.057 & 1.000 & 0.144 & 0.288 & 0.097 \\
\hline
\textbf{1000} & 1.000 & 0.007 & 1.000 & 0.016 & 0.339 & 0.122 \\
\bottomrule
\end{tabular}%
}
\label{F1_compare_Poisson}
\end{table}%

\begin{table}[htbp]
\centering
\caption{F1-score for Poisson regression. $p = 1000$,  high correlation}  \scalebox{0.75}{
\begin{tabular}{|r|r|r|r|r|r|r|r|r|r|r|r|r|}
\toprule
\hline
& \multicolumn{2}{c|}{\textbf{One-step Lasso}} & \multicolumn{2}{c|}{\textbf{Exact}} & \multicolumn{2}{c|}{\textbf{Lasso}} \\
\hline
\multicolumn{1}{|l|}{\textbf{n}} & \multicolumn{1}{c|}{\textbf{Median}} & \multicolumn{1}{c|}{\textbf{Std. Error}} & \multicolumn{1}{c|}{\textbf{Median}} & \multicolumn{1}{c|}{\textbf{Std. Error}} & \multicolumn{1}{c|}{\textbf{Median}} & \multicolumn{1}{c|}{\textbf{Std. Error}} \\
\hline
\textbf{200} & 0.222 & 0.132 & 0.286 & 0.164 & 0.229 & 0.064 \\
\hline
\textbf{300} & 0.293 & 0.160 & 0.556 & 0.189 & 0.265 & 0.070 \\
\hline
\textbf{400} & 0.440 & 0.168 & 0.594 & 0.198 & 0.281 & 0.076 \\
\hline
\textbf{500} & 0.596 & 0.189 & 0.778 & 0.130 & 0.314 & 0.057 \\
\hline
\textbf{1000} & 0.783 & 0.089 & 0.894 & 0.115 & 0.310 & 0.052 \\
\hline
\textbf{1500} & 0.952 & 0.062 & 1.000 & 0.076 & 0.286 & 0.046 \\
\hline
\textbf{2000} & 1.000 & 0.057 & 1.000 & 0.070 & 0.299 & 0.066 \\
\hline
\bottomrule
\end{tabular}%
}
\label{F1_compare_rho09_Poisson}
\end{table}%

\subsection{A breast cancer data application}
We illustrate the method with a high-profile breast cancer data example taken from \cite{geneexpression}. Accurate prediction of distant metastasis development in breast cancer can help guide treatments, and ultimately save life while preserving quality of life. In \cite{geneexpression:2}, using gene expression data, the authors developed a prognostic score based on 70 identified genes for predicting the advent of metastasis of breast cancer cells  in distant organs  within 10 years of diagnostics. In this illustration we re-analyzed the data collected by the same research group in \cite{geneexpression} and we compare their \textsf{70-genes} profile rule with a logistic regression model for predicting the advent of distant metastasis within 5 years.

The data contains 295 women, all younger than 53, with stage I or II breast cancer. Gene intensity measurements of 24496 genes along with 13 clinical variables were collected on the patients. Of the 295 patients, 151 had lymph-node-negative disease, and 144 had lymph-node-positive disease. Ten of the lymph-node-negative, and 120 of the lymph-node-positive had received more aggressive therapy, including chemotherapy, or hormonal therapy, or both. Clearly, the advent of distant metastasis depends on the stage of the cancer when detected, and on initial treatment received. However, despite the cases of lymph-node diseases, most of the cancer cases appear to be at similar stages, and  following \cite{geneexpression}, we will not account for these interactions in the analysis. As response variable $y$, we consider the advent of distant metastasis within five years, as a binary response $\{0, 1\}$. 

To reduce the size of the covariates, we adopted the pre-processing method implemented by \cite{guo2018some}, with one-at-the-time initial logistic regressions that keeps only genes with p-value less than  $0.01$. To this initial set of genes, we then add the $70$ genes identified by \cite{geneexpression}, if they are not already selected by the individual T-tests. This pruning process generates a dataset with $295$ patients and $1083$ genes.

We compare the performance of \textsf{SparseVB}, \textsf{OLAP}, and \textsf{70-genes}, the predictive model of \cite{geneexpression} based on the 70 genes that they have identified. For the comparison we use a 50-fold cross-validation procedure. In each cross-validation replication the test sample size is $30$, and the remaining $265$ samples are used for training. To avoid distortion, when sampling the test set, we require the number of $1$ and $0$ be approximately equal. As a performance metric we compute the $F_1$-score in correctly predicting the outcome variable on the test set. Table \ref{Real_data_performance_F1} shows the mean, median and standard deviation of the $F_1$-score from the 50 cross-validation replications. The results clearly show a better performance of \textsf{OLAP} both in terms of accuracy and stability.

\begin{table}[htbp]
  \centering
  \caption{F1 score of each algorithm for breast cancer data}
  \begin{tabular}{|c|c|c|c|}
    \hline
    {} & {\textbf{Sparse VB}} & {\textbf{70-identifiers}} & {\textbf{OLAP}} \\
    \hline
    \textbf{Mean} & 0.362 & 0.593 & 0.688 \\
    \hline
    \textbf{Median} & 0.509 & 0.609 & 0.688 \\
    \hline
    \textbf{Std. Error} & 0.294 & 0.109 & 0.076 \\
    \hline
  \end{tabular}
  \label{Real_data_performance_F1}
\end{table}

\subsection{Mouse PCR data application}
As another illustration, we analyze the mouse PCR data also previously analyzed in \cite{lan2006,narisetty2018skinny}.  The dataset comprises expression levels of $22575$ genes obtained from $31$ female and $29$ male mice, totaling $60$ arrays. Additionally, the physiological phenotype glycerol-3-phosphate acyltransferase (GPAT) was measured using quantitative real-time PCR. These gene expression and phenotypic data are publicly accessible on the GEO database (http://www.ncbi.nlm.nih.gov/geo; accession number GSE3330).

We want to predict whether a mouse has low GPAT levels given its genetic expression. The level of GPAT in the body is important, as reduced GPAT levels have been linked to decreased hepatic steatosis, a disease commonly associated with obesity.

Similar to \cite{narisetty2018skinny}, we derive the binary response variable based on GPAT levels, defined as $y = I(GPAT < Q(0.4))$, where $Q(0.4)$ represents the $40\%$ quantile of GPAT. And, given the extensive number of genes, we also did a similar gene pruning as in the previous example, by conducting marginal simple logistic regression of the response $y$ against individual genes. But unlike \cite{narisetty2018skinny}, and in order to make our experiment more challenging, we select 500 genes that have the most marginally significant p-values, instead of 99 in the original work. Together with the gender variable, this results in $p=501$ covariates.

For comparison, we apply \textsf{SparseVB}, \textsf{S-Gibbs} and \textsf{OLAP}. Similar to the breast cancer data experiment, we randomly select $30$ different pairs of training and test data sets. To have comparable results to \cite{narisetty2018skinny}, we report in Table \ref{Real_data_2_performance_RMSE} the square root of mean squared error (RMSE) as a measure of performance, which is $RMSE=\sqrt{\frac1n\sum_i^{n} (y_i-\hat \pi_i)^2}$, where $\hat\pi_i$ is the probability predicted by the logistic model. We also report in Table \ref{Real_data_2_performance_F1} the $F_1$ score in correctly predicting the outcome variable on the test set.  The results here again show that \textsf{OLAP} works better than \textsf{S-Gibbs} and \textsf{SparseVB}, in terms of both the accuracy and stability in these measures.

\begin{table}[htbp]
  \centering
  \caption{RMSE for the mouse PCR example}
  \begin{tabular}{|c|c|c|c|}
    \hline
    {} & {\textbf{Sparse VB}} & {\textbf{Skinny Gibbs}} & {\textbf{OLAP}} \\
    \hline
    \textbf{Mean} & 0.4825   &  0.5751&   0.3172 \\
    \hline
    \textbf{Median}  &  0.5000  &  0.5236  &  0.3080 \\
    \hline
    \textbf{Std. Error} &  0.0903  &  0.1279  &  0.0561 \\
    \hline
  \end{tabular}
  \label{Real_data_2_performance_RMSE}
\end{table}

\begin{table}[htbp]
  \centering
  \caption{F1 score for the mouse PCR example}
  \begin{tabular}{|c|c|c|c|}
    \hline
    {} & {\textbf{Sparse VB}} & {\textbf{Skinny Gibbs}} & {\textbf{OLAP}} \\
    \hline
    \textbf{Mean} & 0.4724 &  0.3786   &   0.8430\\
    \hline
    \textbf{Median} &    0.5357  &   0.3636   &  0.8889 \\
    \hline
    \textbf{Std. Error} & 0.3223  &   0.1993  &  0.1806 \\
    \hline

  \end{tabular}
  \label{Real_data_2_performance_F1}
\end{table}

\section{Concluding remarks}\label{sec:conclusion}
Variable selection is an NP-hard problem (\cite{welch:82}), and algorithms that can solve all versions in  polynomial times are unlikely to exist. Therefore, identifying instances of the problem (and corresponding algorithms) that can be solved in polynomial time is of practical importance. Our work in this paper contributes to this literature. We have developed a novel Laplace approximation algorithm that is applicable to a large class of generalized linear models (GLMs) and beyond. The resulting algorithm is fast and accurate, and under mild conditions, we have shown that it leads to a consistent variable selection methodology in the high-dimensional context. Additionally, we have shown that in many cases the mixing time of the resulting Gibbs sampler scales polynomially in $(n,p)$. The simulation results and the real data analysis illustrate the competitiveness of \textsf{OLAP} against some existing high-dimensional variable selection methods. Another advantage of \textsf{OLAP} is that the method has minimal tuning parameter, and as a result can be easily implemented in statistical software.


One important limitation of \textsf{OLAP} is the computational cost of the one-step Newton update. In the current implementation, each component update in the Gibbs sampler is performed at the cost of $n\|\delta\|_0^2$ operations to form the matrix $\tilde{\mathcal{H}}^\delta$, plus $O(\|\delta\|_0^3/3)$ operations to perform its Cholesky factorization. Finding a recursive update to these calculations will significantly improve the speed of the algorithm. Another potentially useful direction for further investigation is the extension of \textsf{OLAP} beyond GLM, to dealing for instance with hierarchical models, or latent variable models.

\appendix
\section{Proof of Theorem \ref{thm:post:contr}}\label{sec:proof:thm:post:contr}
\begin{proof}
We partition the model space $\Delta$ as
\[\Delta = \bigcup_{\delta_0\in\Cset}\;\Delta(\delta_0),\]
where 
\[\Cset\eqdef\{\delta\in\Delta:\;\delta\subseteq\delta_\star\},\;\;\mbox{ and }\;\; \Delta(\delta_0) \eqdef \left\{\delta'\in\Delta:\;\delta_0\subseteq\delta',\;\mbox{ and }\; \min(\delta',\delta_\star)=\delta_0\right\}.\]
We claim that under the conditions of the theorem the following holds true. For all $\delta_0\in\Cset$, and all $M\geq 1$,
\begin{equation}\label{claim:1}
\sum_{\delta\in\Delta(\delta_0)}\;\frac{\check\Pi(\delta\vert\D)}{\check\Pi(\delta_0\vert\D)} \leq 2\;\;\mbox{ and }\;\; \check\Pi\left(\|\delta\|_0\geq s_\star + M\vert \D\right)  \leq 2\left(\frac{1}{p^{\mathsf{u}/2}}\right)^M.
\end{equation}

\medskip
To use the claim in (\ref{claim:1}) to proof the theorem we observe that $\delta\notin\A_j$ means that either $\|\delta\|_0>s_\star + j$, or $\min(\delta,\delta_\star)$ is strictly a sub-model of $\delta_\star$. We can rewrite this statement as
\[\A_j^c =\left\{\delta\in\Delta:\;\|\delta\|_0>s_\star +j\right\}\; \cup\;\bigcup_{\delta_0\in\Cset\setminus\{\delta_\star\}}\;\left\{\delta\in\Delta(\delta_0):\;\|\delta\|_0\leq s_\star +j\right\}.\]
Therefore, using (\ref{claim:1}),
\begin{multline*}
\check\Pi(\A_j^c\vert\D) \leq \frac{2}{p^{\mathsf{u}(j+1)/2}} + \check\Pi(\delta_\star\vert\D)\sum_{\delta_0\in\Cset\setminus\{\delta_\star\}}\;\frac{\check\Pi(\delta_0\vert\D)}{\check\Pi(\delta_\star\vert\D)}\;\sum_{\delta\in\Delta(\delta_0)}\;\frac{\check\Pi(\delta\vert\D)}{\check\Pi(\delta_0\vert\D)}\\
\leq \frac{2}{p^{\mathsf{u}(j+1)/2}} + 2\sum_{\delta_0\in\Cset\setminus\{\delta_\star\}}\;\frac{\check\Pi(\delta_0\vert\D)}{\check\Pi(\delta_\star\vert\D)}.
\end{multline*}
Take $\delta_0\in\Cset$, with $\|\delta_0\|_0 = s_\star -k$, for some $k>0$. Then from (\ref{Pi:check:2}), we have
\[\frac{\check\Pi(\delta_0\vert \D)}{\check\Pi(\delta_\star\vert \D)} = p^{\mathsf{u}k}\exp\left(\bar\ell^{\delta_0}(\check\theta_{\delta_0};\D) - \bar\ell^{\delta_\star}(\check\theta_{\delta_\star};\D)\right).\]
Since the family $\{\check\theta_\delta,\;\delta\in\Delta\}$ is variable selection consistent, $\bar\ell^{\delta_0}(\check\theta^{\delta_0};\D) - \bar\ell^{\delta_\star}(\check\theta_{\delta_\star};\D)\leq -c_2k n$. It follows that
\[\check\Pi(\A_j^c\vert\D) \leq \frac{2}{p^{\mathsf{u}(j+1)/2}} + 2\sum_{k=1}^{s_\star}{s_\star\choose k}\exp\left(\mathsf{u}k\log(p) -c_2k n \right).\]
Hence for $c_2n\geq 2(\mathsf{u}+1)\log(p)$ as assumed in (\ref{cond:ss}), we obtain
\[
\check\Pi(\A_j^c\vert\D) \leq \frac{2}{p^{\mathsf{u}(j+1)/2}} + 2\sum_{k=1}^{s_\star} e^{-c_2 nk/2}\leq \frac{2}{p^{\mathsf{u}(j+1)/2}}  + 4e^{- c_2n/2},\]
which yields the stated inequality.

It remains to prove (\ref{claim:1}). Given $\delta_0\in\Cset$, $\delta\in\Delta(\delta_0)$, and setting $\|\delta\|_0 - \|\delta_0\|_0 = j$, we have
\[\frac{\check\Pi(\delta\vert \D)}{\check\Pi(\delta_0\vert \D)} = \left(\frac{1}{p^{\mathsf{u}}}\right)^{j}\exp\left(\bar\ell^\delta(\check\theta^{\delta};\D) -\bar\ell^{\delta_0}(\check\theta^{\delta_0};D)\right).\]
The variable selection consistency of $\{\check\theta_\delta,\;\delta\in\Delta\}$ yields $\bar\ell^\delta(\check\theta^{\delta};\D) -\bar\ell^{\delta_0}(\check\theta^{\delta_0};\D)\leq c_1j\log(p)$, so that,
\begin{equation}\label{proof:thm:ratio:1}
\frac{\check\Pi(\delta\vert \D)}{\check\Pi(\delta_0\vert \D)} \leq  \exp\left( -\mathsf{u}j\log(p) + c_1j\log(p) \right),\end{equation}
and we conclude that 
\begin{multline*}
\sum_{\delta\in\Delta(\delta_0)}\;\frac{\check\Pi(\delta\vert\D)}{\check\Pi(\delta_0\vert\D)} = \sum_{j\geq 0} \;\;\sum_{\delta\in\Delta(\delta_0):\;\|\delta\|_0=\|\delta_0\|_0 +j}\;\;\frac{\check\Pi(\delta\vert\D)}{\check\Pi(\delta_0\vert\D)}\\
\leq \sum_{j\geq 0} \; {p-s_0\choose j}e^{-(\mathsf{u}-c_1)j\log(p)} \leq \sum_{j\geq 0} \;e^{-(\mathsf{u}-c_1-1)j\log(p)} \leq 2, 
\end{multline*}
provided that $\mathsf{u}\geq 2+c_1$, and $p\geq 2$,
which is satisfied by taking $\mathsf{u}$ as in (\ref{cond:u}). The second part of (\ref{claim:1}) follows a similar argument. Since $\Delta = \cup_{\delta_0\in\Cset}\Delta(\delta_0)$, for any $M\geq 1$, we get
\begin{equation}\label{proof:thm:eq25:2}
\check\Pi\left(\|\delta\|_0\geq s_\star + M\vert \D\right)  = \sum_{\delta_0\in\Cset}\check\Pi(\delta_0\vert\D)\sum_{\delta\in\Delta(\delta_0):\;\|\delta\|_0\geq s_\star +M}\; \frac{\check\Pi(\delta\vert \D)}{\check\Pi(\delta_0\vert \D)}.\end{equation}
Fix $\delta_0\in\Cset$, and set $\|\delta_0\|_0=s_0$. For $k\geq s_\star$, we have
\begin{multline*}
\sum_{\delta\in\Delta(\delta_0):\;\|\delta\|_0\geq k}\; \frac{\check\Pi(\delta\vert \D)}{\check\Pi(\delta_0\vert \D)} = \sum_{j\geq k-s_0}\;\;\sum_{\delta\in\Delta(\delta_0):\;\|\delta\|_0 = s_0 + j}\; \frac{\check\Pi(\delta\vert \D)}{\check\Pi(\delta_0\vert \D)}\\
\leq \sum_{j\geq k-s_0}\;{p-s_0 \choose j} \exp\left(-\mathsf{u}j\log(p)+c_1j\log(p) \right)\\
\leq \sum_{j\geq k-s_0}\;\exp\left(-\left(\mathsf{u}-c_1 -1\right)j\log(p)\right).
\end{multline*}
Hence for $\mathsf{u}/2\geq (c_1 +1)$ as assumed in (\ref{cond:ss}) we obtain for $k> s_0$, that
\[\sum_{\delta\in\Delta(\delta_0):\;\|\delta\|_0\geq k}\; \frac{\check\Pi(\delta\vert \D)}{\check\Pi(\delta_0\vert \D)}\leq \sum_{j\geq k-s_0}\;\left(\frac{1}{p^{\mathsf{u}/2}}\right)^j \leq 2 \left(\frac{1}{p^{\mathsf{u}/2}}\right)^{k-s_0}.\]
Hence, the last display together with (\ref{proof:thm:eq25:2}) implies that  for any $M\geq 1$, we have
\[\check\Pi\left(\|\delta\|_0\geq s_\star + M\vert \D\right) \leq  2\sum_{\delta_0\in\Cset}\check\Pi(\delta_0\vert\D) \left(\frac{1}{p^{\mathsf{u}/2}}\right)^{M}\leq 2\left(\frac{1}{p^{\mathsf{u}/2}}\right)^M.\]

\end{proof}

\section{Proof of Theorem \ref{lem:post:consis}}\label{sec:proof:post:consis}

  \begin{proof}
    Let $q^\delta(\cdot;\D)$ denote the quadratic approximation of $\bar\ell^\delta(\cdot;\D)$ around $\tilde\theta^\delta$. Specifically, for $u\in\rset^{\|\delta\|_0}$,
    \[q^\delta(u;\D) \eqdef \bar\ell^\delta(\tilde\theta^{\delta};\D) + \pscal{\tilde\G_\delta}{u- \tilde\theta^{\delta}} -\frac{1}{2}(u- \tilde\theta^{\delta})' \tilde \H_\delta (u- \tilde\theta^{\delta}),\]
    and
    \[\eta^\delta(u;\D)\eqdef \bar\ell^\delta(u;\D) - q^\delta(u;\D).\]
    An important step in the argument is to upper bound the remainder $|\eta^\delta(u;\D)|$ for $u$ close to $\theta_\star$. This can be easily done in the setting where the third derivative of $\psi$ is uniformly bounded. However this will leave out the Poisson model. Following (\cite{bach10}), we use the self-concordance assumption  in H\ref{H:basic}-(3) to handle a larger class of link functions $\psi$. Specifically, under assumption H\ref{H:basic}, the following holds. We can find a constant $C$ that depends only on $c_3$, $\bar\kappa$ and $b$ in H\ref{H:basic} such that for all $\delta\subseteq\delta_\star$,
    \begin{equation}\label{claim:thm:behav:ll}
      \left|\eta^\delta(u;\D)\right| \leq C ne^{c_3b\|u-\tilde\theta^\delta\|_1} \times \|u-\tilde\theta^\delta\|_1 \times \|u-\tilde\theta^\delta\|_2^2,\;\;\;u\in\rset^{\|\delta\|_0}.
    \end{equation}
    This claim is proved below. Let us assume for the time being that (\ref{claim:thm:behav:ll}) holds, and fix $\delta,\delta_0\in\Delta$. Since $\hat\theta^{\delta}$ maximizes $\bar\ell^\delta(\cdot;\D)$, we have 
    \begin{equation}\label{ll:lb1}
      \bar\ell^\delta(\check\theta^{\delta};\D) \leq \bar\ell^\delta(\hat\theta^{\delta};\D).\end{equation}
    But since $\check\theta^{\delta}$ maximizes $q^\delta(\cdot;\D)$, we have
    \begin{multline}\label{ll:ub1}
      \bar\ell^\delta(\check\theta^{\delta};\D) = q^\delta(\check\theta^{\delta};\D) + \eta^\delta(\check\theta^{\delta};\D) \geq  q^\delta(\hat\theta^{\delta};\D)  + \eta^\delta(\check\theta^{\delta};\D)\\
      = \bar\ell^\delta(\hat\theta^{\delta};\D)  + \eta^\delta(\check\theta^{\delta};\D) - \eta^\delta(\hat\theta^{\delta};\D).
    \end{multline}
    We combine (\ref{ll:lb1}) and (\ref{ll:ub1}) to conclude that
    \begin{equation}\label{ll:final:b}
      \bar\ell^\delta(\check\theta^{\delta};\D) - \bar\ell^{\delta_0}(\check\theta^{\delta_0};\D) \leq \left[\bar\ell^\delta(\hat\theta^{\delta};\D) - \bar\ell^{\delta_0}(\hat\theta^{\delta_0};\D)\right] + |\eta^{\delta_0}(\check\theta^{\delta_0};\D)|  + |\eta^{\delta_0}(\hat\theta^{\delta_0};\D)|.
    \end{equation}
    Hence, proving the theorem boils down to showing that there exists a constant $C<\infty$ such that for all $\delta_0\subseteq\delta_\star$,
    \begin{equation}\label{unif:bound}
      |\eta^{\delta_0}(\check\theta^{\delta_0};\D)|  + |\eta^{\delta_0}(\hat\theta^{\delta_0};\D)| \leq C\log(p).\end{equation}
    
    Indeed, if the MLE family $\{\hat\theta^\delta,\;\delta\in\Delta\}$ is variable selection consistent with constant $c_1,c_2$, say, then  for $\delta_0\subseteq\delta$, and $\min(\delta,\delta_\star)=\delta_0$, we can conclude from the above argument that
    \[\bar\ell^\delta(\check\theta^{\delta};\D) - \bar\ell^{\delta_0}(\check\theta^{\delta_0};\D) \leq c_1 \left(\|\delta\|_0 -\|\delta_0\|_0\right)\log(p) + C\log(p)\leq (c_1 + C) \left(\|\delta\|_0 -\|\delta_0\|_0\right)\log(p).\]
    Similarly, for $\delta_0\subseteq \delta_\star$, applying again (\ref{ll:final:b}) with $\delta_0\leftarrow \delta_\star$, and $\delta\leftarrow \delta_0$, we get
    \[\bar\ell^{\delta_\star}(\check\theta^{\delta_\star};\D) - \bar\ell^{\delta_0}(\check\theta^{\delta_0};\D) \geq c_2 \left(\|\delta_\star\|_0 -\|\delta_0\|_0 \right)n - C\log(p) \geq \frac{c_2}{2} \left(\|\delta_\star\|_0 -\|\delta_0\|_0 \right)n,\]
    for $c_2n\geq 2C\log(p)$. This establishes the theorem. It remains to prove (\ref{unif:bound}) and (\ref{claim:thm:behav:ll}). 
    \medskip

    \paragraph{\underline{\texttt{Proof of (\ref{claim:thm:behav:ll})}}} 
    Given, $u,h\in\rset$, with $h\neq 0$, define 
    \[g(t) = \psi(u+ t h),\;\;t\in [0,1].\]
    H\ref{H:basic}-(3) implies that $|g^{'''}(t)| \leq c_3 |h|g^{''}(t)$, which in turn implies that for all $t\in [0,1]$,
    \[-c_3|h|\leq \frac{\rmd \log g^{''}(t)}{\rmd t}\leq c_3 |h|.\] 
    Integrating these two inequalities thrice we obtain for all $t\in [0,1]$,
    \begin{equation}\label{self:conc:1}
      g^{''}(0)e^{-c_3|h|t} \leq g^{''}(t) \leq g^{''}(0)e^{c_3|h|t},
    \end{equation}
    \begin{equation}\label{self:conc:2}
      g^{''}(0)\times \frac{1  -e^{-c_3|h|t}}{c_3|h|} \leq g^{'}(t) - g^{'}(0) \leq g^{''}(0)\times \frac{e^{c_3|h|t}  -1 }{c_3|h|},
    \end{equation}
    and
    \begin{equation*}
      -g^{''}(0)\times \frac{1 - c_3|h|t -e^{-c_3|h|t}}{c_3^2h^2} \leq g(t) - g(0) -tg^{'}(0) \leq g^{''}(0)\times \frac{e^{c_3|h|t} -c_3|h|t -1 }{c_3^2h^2}.
    \end{equation*}
    Subtracting $g^{''}(0)/t^2$ from all sides we get
    \begin{multline*}
      -g^{''}(0)\times \frac{1 - c_3|h|t + (c_3|h|t)^2/2-e^{-c_3|h|t}}{c_3^2h^2} \leq g(t) - g(0) -tg^{'}(0) -\frac{t^2}{2}g^{''}(0) \\
      \leq g^{''}(0)\times \frac{e^{c_3|h|t} -(c_3|h|t)^2/2 -c_3|h|t -1 }{c_3^2h^2}.
    \end{multline*}
    We apply this with $t=1$, and we use a Taylor of $e^{x}$ and $e^{-x}$ to the third order to conclude that
    \begin{equation}\label{self:conc:3}
      \left|\psi(u+h) - \psi(u) -\psi'(u)h -\frac{h^2}{2}\psi^{''}(u)\right|\leq \frac{c_3}{6}|h|^3 e^{c_3|h|}\psi^{''}(u).\end{equation}
    Since for all $u\in\rset^{\delta\|_0}$,
    \begin{multline*}\eta^\delta(u;\D) = \sum_{i=1}^n \left[\psi\left(\pscal{u}{{\bf x}_i}\right) - \psi\left(\pscal{\tilde\theta^\delta}{{\bf x}_i}\right) - \psi'\left(\pscal{\tilde\theta^\delta}{{\bf x}_i}\right)\pscal{u-\tilde\theta^\delta}{{\bf x}_i}\right.\\
      -\left. \frac{1}{2}\psi^{''}\left(\pscal{\tilde\theta^\delta}{{\bf x}_i}\right)\pscal{u-\tilde\theta^\delta}{{\bf x}_i}^2\right].\end{multline*}
    We conclude with (\ref{self:conc:3}) that
    \begin{multline*}
      \left|\eta^\delta(u;\D)\right| \leq \frac{c_3b}{6}\|u-\tilde\theta^\delta\|_1e^{c_3b\|u-\tilde\theta^\delta\|_1}\sum_{i=1}^n\psi^{''}\left(\pscal{\tilde\theta^\delta}{{\bf x}_i}\right)\pscal{u-\tilde\theta^\delta}{{\bf x}_i}^2\\
      \leq \frac{c_3b\bar\kappa n }{6}e^{c_3b\|u-\tilde\theta^\delta\|_1} \|u-\tilde\theta^\delta\|_1 \times \|u-\tilde\theta^\delta\|_2^2,
    \end{multline*}
    which is the claim (\ref{claim:thm:behav:ll}).
    
    \paragraph{\underline{\texttt{Proof of (\ref{unif:bound})}}} 
    Fix $\delta_0\subseteq\delta_\star$. Since $\hat\theta^{\delta_0} -\tilde\theta^{\delta_0} = \hat\theta^{\delta_0} -\theta_\star^{\delta_0} +\theta_\star^{\delta_0} -  \tilde\theta^{\delta_0}$, and $\delta_0\subseteq\delta_\star$, we get, using (\ref{rate:estimators}),
    \[\|\hat\theta^{\delta_0} -\tilde\theta^{\delta_0}\|_2 \leq \|\hat\theta^{\delta_\star} - \theta_\star\|_2 + \|\wtilde\theta - \theta_\star\|_2 \leq C\sqrt{\frac{s_\star\log(p)}{n}}.\]
    Hence
    \[\|\hat\theta^{\delta_0} -\tilde\theta^{\delta_0}\|_1 \leq C\sqrt{\frac{s_\star^2\log(p)}{n}}\leq C',\]
    for $n\geq s_\star^2\log(p)$. As a result, using (\ref{claim:thm:behav:ll}), we conclude that
    \[|\eta^{\delta_0}(\check\theta^{\delta_0};\D)| \leq Cn \sqrt{\frac{s_\star^2\log(p)}{n}} \times \frac{s_\star\log(p)}{n}\leq C'\log(p),\]
    under the sample size condition (\ref{cond:ss:2}). By the definition of $\check\theta^\delta$ in (\ref{theta:check}),
    \[\check\theta^{\delta_0} -\tilde\theta^{\delta_0} = (\tilde\H^{\delta_0})^{-1}\tilde \G^{\delta_0},\]
    and under H\ref{H:basic}-(4), the smallest eigenvalue of $\tilde\H^{\delta_0}$ is at least $\underline{\kappa} n$. Therefore, and using the expression of the gradient $\G^\delta$, we have
    \[\|\check\theta^{\delta_0} -\tilde\theta^{\delta_0}\|_2 \leq \frac{1}{\underline{\kappa} n} \left(\left\|\left[\G_\star\right]_{\delta_0}\right\|_2 +\left\|\tilde \G^{\delta_0} - \left[\G_\star\right]_{\delta_0}\right\|_2\right),\]
    where we define $\G_\star\eqdef\nabla\bar\ell^\delta(u;\D)\vert_{u=\theta_\star}$. 
    According to H\ref{H:basic}-(2), $\|\G_\star\|_\infty\leq M_0\sqrt{n\log(p)}$. Hence 
    \[\frac{\left\|\left[\G_\star\right]_{\delta_0}\right\|_2}{\underline{\kappa} n} \leq C \sqrt{\frac{\|\delta_0\|_0 \log(p)}{n}} \leq C \sqrt{\frac{s_\star \log(p)}{n}}.\]
    On the other hand, a comparison between the $j$-th component of $\tilde \G^{\delta_0}$  and $\left[\G_\star\right]_{\delta_0}$ gives
    \begin{multline*}
      (\tilde \G^{\delta_0})_j - \left[\G_\star\right]_j = -\sum_{i=1}^n \left(\psi'(\pscal{\tilde\theta^{\delta_0}}{{\bf x}_i}) - \psi'(\pscal{\theta_\star}{{\bf x}_i})\right)[{\bf x}_i]_{j} \\
      = -\sum_{k:\;\delta_{\star,k}=1}\; (\tilde\theta_k - \theta_{\star,k}) \sum_{i=1}^n \psi{''}(\pscal{\vartheta}{{\bf x}_i}){\bf x}_{ij}{\bf x}_{ik},\end{multline*}
    for some $\vartheta$ on the segment between $\tilde\theta^{\delta_0}$ and $\theta_\star$. We can then appeal to H\ref{H:basic}-(4-5) to conclude that
    \[\left|(\tilde \G^{\delta_0})_j - \left[\G_\star\right]_j\right| \leq \bar\kappa n |\tilde\theta_j - \theta_{\star,j}| + M \|\wtilde \theta - \theta_\star\|_1 \sqrt{n\log(p)} \leq \bar\kappa n |\tilde\theta_j - \theta_{\star,j}| + C s_\star \log(p).\]
    We deduce that
    \[\frac{1}{\underline{\kappa} n} \left\|\tilde \G^{\delta_0} - \left[\G_\star\right]_{\delta_0}\right\|_2 \leq \frac{Cn}{\underline{\kappa}n} \sqrt{\frac{s_\star \log(p)}{n}} + \frac{C s_\star^{3/2}\log(p)}{n} \leq C' \sqrt{\frac{s_\star \log(p)}{n}}.\]
    Therefore, (\ref{claim:thm:behav:ll}), and a similar argument yields
    \[|\eta^{\delta_0}(\hat\theta^{\delta_0};\D)| \leq Cn \sqrt{\frac{s_\star^2\log(p)}{n}} \times \frac{s_\star\log(p)}{n}\leq C'\log(p),\]
    which concludes the proof of (\ref{unif:bound}).
    
  \end{proof}

\section{Proof of Theorem \ref{thm:mix}}\label{sec:proof:thm:mix}
When $J=1$, Algorithm \ref{algo:1} is also known as a random scan Gibbs sampler. It generates a reversible and positive Markov chain with invariant distribution $\check\Pi$ that fits in the framework presented in Section \ref{sec:mix:mc}. We denote $K$ its transition kernel.  We prove the result by applying Lemma \ref{lem:key:sprofile} and Proposition \ref{lb:cp}. The initial distribution of Algorithm \ref{algo:1} has a density $f_0$ with respect  to $\check\Pi$ given by: $f_0(\delta) = 1/\check\Pi(\delta^{(0)}\vert \D)$, if $\delta=\delta^{(0)}$, and zero everywhere else. Hence
\[\|f_0\|_\infty = \frac{1}{\check\Pi(\delta^{(0)}\vert \D)}\leq p^\alpha.\]
Fix $\Delta_0\subseteq\Delta$ to be determined later such that $\check\Pi(\Delta_0\vert\D) \geq 15/16$, and set $\zeta = 8(1-\check\Pi(\Delta_0\vert\D))$. If we can build canonical paths on $\Delta_0$ such that H\ref{H2} holds, then by Proposition \ref{lb:cp}, we get $\lambda_\zeta(K) \geq 1/\mathsf{m}(\Delta_0)$. And since $\log(1-x)\leq -x$ for all $x\in (0,1]$, it follows that for  $N\geq 1$ large enough such that 
\[N \geq 2\mathsf{m}(\Delta_0) \log\left(\frac{\|f_0\|_\infty^2}{\zeta_0^2}\right),\]
we have 
\[\left(1-\frac{\lambda_\zeta(K)}{2}\right)^N = \exp\left(N\log\left(1-\frac{\lambda_\zeta(K)}{2}\right)\right) \leq \exp\left(-\frac{N}{2\mathsf{m}(\Delta_0)}\right) \leq \frac{\zeta_0^2}{\|f_0\|_\infty^2}.\] 
Therefore by (\ref{lem:key:eq:pos}) of Lemma \ref{lem:key:sprofile}, 
\[\|K^N(\delta^{(0)},\cdot) -\check\Pi\|_\tv^2 \leq \max\left(\zeta\|f_0\|_\infty^2,\left(1-\frac{\lambda_\zeta(K)}{2}\right)^N\|f_0\|_\infty^2\right)\leq \max\left(\zeta\|f_0\|_\infty^2,\zeta_0^2\right).\]
We apply this result with the choice $\Delta_0=\{\delta\in\Delta:\;\|\delta\|_0\leq s_\star +J_0\}$. By (\ref{claim:1}), 
\[\zeta \|f_0\|_\infty^2 = 8\|f_0\|_\infty^2 (1-\check\Pi(\Delta_0\vert\D) \leq \frac{16\|f_0\|_\infty^2}{p^{\mathsf{u}(J_0+1)/2}}\leq \frac{16}{p^{\mathsf{u}/2}}\frac{\|f_0\|_\infty^2}{p^{\mathsf{u}J_0/2}} \leq \frac{16}{p^{\mathsf{u}/2}},\]
by taking  $J_0 = 4\alpha/\mathsf{u}$
Hence, for
\[N \geq 4\mathsf{m}(\Delta_0) \left(\log(1/\zeta_0) + \alpha\log(p)\right),\]
\begin{equation}\label{proof:mix:eq:fin}
\|K^N(\delta^{(0)},\cdot) -\check\Pi\|_\tv \leq \max\left(\frac{4}{p^{\mathsf{u}/4}},\zeta_0\right).\end{equation}
We now use the canonical path argument to upper bound the term $\mathsf{m}(\Delta_0)$. First, we build a graph on $\Delta_0$ by putting an edge between $\delta$ and $\delta'$ if  $\|\delta-\delta'\|_0 =1$. Clear H\ref{H2} holds. We build the canonical paths by removing or adding variables as follow. For any $\delta$ we first build a path between $\delta$ and $\delta_\star$.  First, we set to zero (1 term at the time and in their decreasing index order) the components $j$ of $\delta$ for which $\delta_j=1$ and $\delta_{\star j} =0$.  We then reach $\underline{\delta}\eqdef \min(\delta,\delta_\star)$. This corresponds to removing one-by-one non-relevant variables from the model $\delta$. Then we set to one ( one at the time, and in their increeasing index order) the component $j$ of $\underline{\delta}$ for which $\underline{\delta}_j=0$, and $\delta_{\star j}=1$. We then reach $\delta_\star$. This corresponds to adding one-by-one relevant variables not already contained in $\underline{\delta}$.

Given  two arbitrary points $\delta$ and $\delta'$, we build the canonical path between them as follows. Let $\vartheta$ denote the point where the canonical path from $\delta$ to $\delta_\star$ and the canonical path from  $\delta'$  to $\delta_\star$ meet for the first time. The canonical path from $\delta$ to $\delta'$ is obtained by following the canonical path from $\delta$ to $\delta_\star$ until $\vartheta$ is reached, then we follow the canonical path from $\delta'$ to $\delta_\star$ in reverse direction starting from $\vartheta$ until $\delta'$. We have
\begin{equation}\label{eq:path:len}
|\gamma_{\delta,\delta'}|\leq 2\left(s_\star + J_0\right).\end{equation}
Let $x,y$ denote generic elements of $\Delta_0$. Fix an edge  $e=(\delta',\delta)$.  Let $\Lambda(e) $ be the set of all elements of $\Delta_0$  whose path to $\delta_\star$ go through $e$.  If $\gamma_{xy}\ni e$ then $x\in\Lambda(e)$ or $y\in\Lambda(e)$ (but not both). Therefore,
\begin{multline*}
\sum_{\gamma_{xy}:\;\gamma_{xy}\ni e}\frac{|\gamma_{xy}|\check\Pi(x\vert \D)\check\Pi(y\vert \D)}{\check\Pi(\delta'\vert \D)K(\delta',\delta)}  \leq 2\sum_{x\in\Lambda(e)}\sum_{y\in\Xset_0}\frac{|\gamma_{xy}|\check\Pi(x\vert \D)\check\Pi(y\vert \D)}{\check\Pi(\delta'\vert \D)K(\delta',\delta)} \\
\leq 2\sum_{x\in\Lambda(e)}\frac{|\gamma_{xy}|\check\Pi(x\vert \D)}{\check\Pi(\delta'\vert \D)K(\delta',\delta)}.\end{multline*}
Using the last display with (\ref{eq:path:len}), we  conclude that
\[\mathsf{m}(\Delta_0) \leq 4\left(s_\star + J_0\right)\;\max_{(\delta',\delta)\in\e_0}\frac{1}{K(\delta',\delta)}\sum_{x\in\Lambda(e)}\frac{\Pi(x\vert \D)}{\Pi(\delta'\vert \D)}.\]

\paragraph{\underline{\texttt{Case 1}}:\; $\delta=\delta^{(j,0)}$, and $\delta'=\delta^{(j,1)}$, for some $j$ such that $\delta_{\star j}=0$} In this case, 
\begin{multline*}K(\delta', \delta) = \frac{1}{p} \left(1- q_j(\delta')\right)\\
= \frac{1}{p}\left(1 +\exp(-\mathsf{u}\log(p) + \bar\ell^{\delta^{(j,1)}}(\check\theta_{\delta^{(j,1)}};\D) - \bar\ell^{\delta^{(j,0)}}(\check\theta_{\delta^{(j,0)}};\D))\right)^{-1} \\
\\ \geq \frac{1}{p}\left(1 +\frac{1}{p^{\mathsf{u}/2}}\right)^{-1}\geq\frac{1}{p}\left(1 -\frac{1}{p^{\mathsf{u}/2}}\right)\geq \frac{1}{2p},
\end{multline*}
where we use (\ref{msc:eq1}) to claim that $\bar\ell^{\delta^{(j,1)}}(\check\theta_{\delta^{(j,1)}};\D) - \bar\ell^{\delta^{(j,0)}}(\check\theta_{\delta^{(j,0)}};\D) \leq c_1\log(p)$, and  for $\mathsf{u}>2c_1$, and $p^{\mathsf{u}/2}\geq 2$.  Furthermore, in this case, note that $\delta'\supset\delta$, and  differs from $\delta$ only at  some $j$ such that $\delta_{\star j}=0$. Therefore, $\Lambda(e)=\Delta(\delta')$, that is the elements $\vartheta\in\Delta_{s_\star +J_0}$ contains $\delta'$ and differs from $\delta'$ only at components at which $\delta_{\star k}=0$. Therefore, as seen in the proof of Theorem \ref{thm:post:contr},
\[
\sum_{x\in\Lambda(e)}\frac{\check\Pi(x\vert \D)}{\check\Pi(\delta'\vert \D)}  \leq 2.\]
and it follows that
\[\frac{1}{K(\delta',\delta)}\sum_{x\in\Lambda(e)}\frac{\Pi(x\vert \D)}{\Pi(\delta'\vert \D)} \leq 4p.\]

\paragraph{\underline{\texttt{Case 2}}:\;:\; $\delta'\subset\delta_\star$, $\delta'=\delta^{(j,0)}$, and $\delta=\delta^{(j,1)}$, for some $j$ such that $\delta_{\star j}=1$}In this case, using (\ref{msc:eq2}),
\begin{multline*}
K(\delta', \delta) = \frac{q_j(\delta')}{p} = \frac{1}{p}\left(1 +\exp(\mathsf{u}\log(p) - \bar\ell^{\delta^{(j,1)}}(\check\theta_{\delta^{(j,1)}};\D) + \bar\ell^{\delta^{(j,0)}}(\check\theta_{\delta^{(j,0)}};\D))\right)^{-1} \\
\\ \geq \frac{1}{p}\left(1 +\exp\left(\mathsf{u}\log(p) - c_2 n\right)\right)^{-1}\geq\frac{1}{p}\left(1 + \frac{1}{e^{c_2n/2}}\right)\geq \frac{1}{2p},
\end{multline*}
for $c_2n \geq 2\mathsf{u}\log(p)$. In this configuration, we see that
\[\Lambda(e)=\bigcup_{\vartheta\subseteq\delta'}\Delta(\vartheta),\]
therefore,
\[
\sum_{x\in\Lambda(e)}\frac{\check\Pi(x\vert \D)}{\check\Pi(\delta'\vert \D)} = \sum_{\vartheta\subseteq\delta'} \frac{\check\Pi(\vartheta\vert \D)}{\check\Pi(\delta'\vert \D)}\sum_{x\in\Delta(\vartheta)} \frac{\check\Pi(x\vert \D)}{\check\Pi(\vartheta\vert \D)}\leq 2 \sum_{\vartheta\subseteq\delta'} \frac{\check\Pi(\vartheta\vert \D)}{\check\Pi(\delta'\vert \D)}.\]
We proceed similarly as in the proof of Theorem \ref{thm:post:contr} to show that
\[\sum_{\vartheta\subseteq\delta'} \frac{\check\Pi(\vartheta\vert \D)}{\check\Pi(\delta'\vert \D)} \leq 1 + 4 e^{-c_2n/2} \leq 2,\]
for $c_2n\geq 2(\mathsf{u}+1)\log(p)$. We conclude that for some absolute constant $C$. 
\[\mathsf{m}(\Delta_0) \leq C (s_\star +J_0)p.\]
We combine this bound with (\ref{proof:mix:eq:fin}) and $J_0 = 4\alpha/\mathsf{u}$, to conclude that
\[\|K^N(\delta^{(0)},\cdot) -\check\Pi\|_\tv \leq \max\left(\frac{4}{p^{\mathsf{u}/4}},\zeta_0\right),\;\\
\mbox{ for }\; N\geq C\left(s_\star + \frac{4\alpha}{\mathsf{u}}\right)\left(\log(1/\zeta_0) + \alpha\log(p)\right)p,\]
for some absolute constant $C$. This concludes the proof.

\section{Proof of the results of Section \ref{sec:mix:mc}}
\subsection{Proof of Lemma \ref{lem:key:sprofile}}\label{sec:proof:lem:key}
By reversibility, for all $j\geq 1$, the density of $\pi_0 K^j$ with respect to $\pi$ is $f_j=K^j f_0$. Therefore
\[\|\pi_0 K^j-\pi\|_\tv =\int|f_j(x) -1|\pi(\rmd x) \leq \sqrt{\textsf{Var}_\pi(f_j)}.\]
Take $f\in L^2(\pi)$. Since $\pi(f) = \pi(Kf)$, and $\textsf{Var}_\pi(f) = \pscal{f}{f}_\pi -\pi(f)^2$, we have
\[
\textsf{Var}_\pi(Kf) - \textsf{Var}_\pi(f) = \pscal{Kf}{Kf}_\pi - \pscal{f}{f}_\pi =  \pscal{f}{K_\star K f}_\pi - \pscal{f}{f}_\pi =-\e_{K^\star K}(f,f).\]
If $K$ is positive (it is therefore $\pi$- reversible), then $K^\star K = K^2$, and $K$ admits a square root: there exists a bounded $\pi$-reversible operator $S$ such that $S^2=K$, and $S$ commutes with $K$. Furthermore, with $\mathbb{I}$ denoting the identity operator, $\mathbb{I}-K$ is also a positive operator: $\mathbb{I}-K$ is clearly $\pi$-reversible, and $\pscal{f}{(\mathbb{I}-K)f}_\pi = \|f\|_2^2 - \pscal{f}{Kf}_\pi \geq 0$, using the fact that the operator norm of $K$ is smaller of equal to one. Hence 
\[ \pscal{f}{Kf}_\pi - \pscal{f}{K^\star K f}_\pi = \pscal{f}{(K-K^2)f}_\pi = \pscal{f}{S(\mathbb{I}-K)Sf}_\pi = \pscal{Sf}{(\mathbb{I}-K)Sf}_\pi  \geq 0.\] Therefore when $K$ is positive we can have
\[
\textsf{Var}_\pi(Kf) \leq \textsf{Var}_\pi(f) -\e_{K}(f,f).\]
In particular, for all $j\geq 1$,
\begin{equation}\label{keylem:eq2}
\textsf{Var}_\pi(f_{j}) \leq  \textsf{Var}_\pi(f_{j-1})  -\e_{K}(f_{j-1},f_{j-1}).
\end{equation}
First, we observe that the last display implies that $\{\textsf{Var}_\pi(f_k),\;k\geq 0\}$ is non-increasing. Fix $j\geq 1$, and suppose now that $\textsf{Var}_\pi(f_j)> \zeta \|f_0\|_\infty^2$. Since the sequence $\{\textsf{Var}_\pi(f_k),\;k\geq 0\}$ is non-increasing, for all $0\leq i\leq j$, $\textsf{Var}_\pi(f_i)> \zeta  \|f_0\|_\infty^2$. Note also that $\|f_k\|_\infty\leq \|f_0\|_\infty$ for all $k\geq 0$. Therefore, for all $1\leq i\leq j$,
\begin{eqnarray*}
\textsf{Var}_\pi(f_{i}) & \leq &  \textsf{Var}_\pi(f_{i-1})  -\e_{K}\left(\frac{f_{i-1}}{\|f_0\|_\infty},\frac{f_{i-1}}{\|f_0\|_\infty}\right)\|f_0\|_\infty^2\\
& \leq & \textsf{Var}_\pi(f_{i-1}) - \|f_0\|_\infty^2 \lambda_\zeta(K)\left(\textsf{Var}_\pi\left(\frac{f_{i-1}}{\|f_0\|_\infty}\right)  - \frac{\zeta}{2}\right),\\
& \leq & \textsf{Var}_\pi(f_{i-1}) -\frac{\lambda_\zeta(K)}{2}\textsf{Var}_\pi(f_{i-1})\\
& \leq  & \left(1 - \frac{\lambda_\zeta(K)}{2}\right)^i\textsf{Var}_\pi(f_{0}).
\end{eqnarray*}
We conclude that for all $j\geq 1$,
\[\textsf{Var}_\pi(f_{j}) \leq \max\left[ \zeta\|f_0\|_\infty^2, \left(1 - \frac{\lambda_\zeta(K)}{2}\right)^j\textsf{Var}_\pi(f_{0})\right].\]
This ends the proof.

\subsection{Proof of Theorem \ref{lem:cheeger} }\label{sec:proof:lem:cheeger}
The proof follows the same argument in the proof of Cheeger's inequality due to \cite{lawler:sokal:88}. Fix a measurable set $A$ such that $\epsilon<\pi(A)<1-\epsilon$, and set
\[f_A(\cdot) = \pi(A)\textbf{1}_{A^c}(\cdot) - \pi(A^c) \textbf{1}_A(\cdot).\]
Clearly $\pi(f_A)=0$, and we check that $(I-K)f_A(\cdot) = \textbf{1}_{A^c}(\cdot) - K(\cdot,A^c)$, and 
\[\e_K(f_A,f_A) = \pscal{f_A}{(I-K)f_A}_\pi = \int_A \pi(\rmd x)K(x,A^c).\]
Also we have
\[\textsf{Var}_\pi(f_A) = \pi(A)\pi(A^c) = \pi(A)(1-\pi(A)).\] 
We check that for  all $\epsilon\in [0,1/2)$,
\[\frac{1 -\sqrt{1-2\epsilon}}{2}<\epsilon < 1-\epsilon < \frac{1+\sqrt{1-2\epsilon}}{2},\]
which implies that for all $\epsilon \leq x\leq 1-\epsilon$, $x(1-x)>\epsilon/2$. We use this to conclude that $\textsf{Var}_\pi(f_A) =\pi(A)(1-\pi(A)) >\epsilon/2$.  Furthermore,  
\begin{multline*}
(\pi(A)-\epsilon)(\pi(A^c)-\epsilon)= \pi(A)\pi(A^c) -\epsilon +\epsilon^2 \leq \pi(A)\pi(A^c) -\frac{\epsilon}{2} \leq  \textsf{Var}_\pi(f_A) -\frac{\epsilon}{4}.\end{multline*}
It follows that
\[\frac{\int_A\pi(\rmd x)K(x,A^c)}{(\pi(A)-\epsilon)(\pi(A^c)-\epsilon)} \geq \frac{\e_K(f_A,f_A)}{\textsf{Var}_\pi(f_A) -\frac{\epsilon}{4}} \geq  \lambda_{\epsilon/2}(K).\]
We use this to conclude as claimed that
\[\lambda_{\epsilon/2}(K) \leq \Phi_\epsilon(K).\]

Take $f\in L^2(\pi)$ with $\|f\|_\infty\leq  1$, and  $\textsf{Var}_\pi(f) >\epsilon$. We aim to lower bound the term $\e_K(f,f)/(\textsf{Var}_\pi(f) - \epsilon/2)$. Given $c\in\rset$, let $f_c\eqdef f-c$, and let $\bar K(\rmd x,\rmd y) \eqdef \pi(\rmd x) K(x,\rmd y)$. By the Cauchy-Schwarz inequality 
\begin{multline*}
\left(\int \bar K(\rmd x,\rmd y) \abs{f_c^2(y )- f_c^2(x)}\right)^2 \\
\leq \int \bar K(\rmd x,\rmd y) (f_c(y )- f_c(x))^2 \times  \int \bar K(\rmd x,\rmd y) (f_c(y ) + f_c(x))^2 \\
\leq 4\pi(f_c^2) \int \bar K(\rmd x,\rmd y) (f(y )- f(x))^2.
\end{multline*}
We deduce that
\[\sqrt{\e_K(f,f)} \geq \frac{\int \bar K(\rmd x,\rmd y) \abs{f_c^2(y )- f_c^2(x)}}{\sqrt{8\pi(f_c^2)}}.\]
By reversibility,
\[\int \bar K(\rmd x,\rmd y) \abs{f_c^2(y )- f_c^2(x)} =2 \int_{\{(x,y):\; f_c^2(y)>f_c^2(x)\}} \bar K(\rmd x,\rmd y)\left(f_c^2(y) - f_c^2(x)\right).\]
For $\alpha\geq 0$, we set $A_\alpha \eqdef\{x:\; f_c^2(x)>\alpha\}$, and $\bar A_\alpha \eqdef\{x:\; f_c^2(x)\leq \alpha\}$ the complement of $A_\alpha$. Let 
\[I_c \eqdef \left[\texttt{ess-inf}(f_c^2),\; 
\texttt{ess-sup}(f_c^2)\right],\]
We note that $\pi(A_\alpha) =0$ for $\alpha>\texttt{ess-sup}(f_c^2)$, and for $\alpha<  \texttt{ess-inf}(f_c^2)$, $\pi(\bar A_\alpha)=0$. Hence in both cases $\int_{A_\alpha}\pi(\rmd x)
K(x,\bar A_\alpha) = 0$. We also observe that for all $x,y\in\Xset$,
\[\int \textbf{1}_{A_\alpha}(x) \textbf{1}_{\bar A_\alpha}(y)\rmd \alpha = (f_c^2(x) - f_c^2(y)) \textbf{1}_{\{f_c^2(x) > f_c^2(y)\}}.\]
Using this and Fubini's theorem,
\begin{multline*}
\int \bar K(\rmd x,\rmd y) \abs{f_c^2(y )- f_c^2(x)}  = 2\int\rmd \alpha \int \bar K(\rmd x,\rmd y) \textbf{1}_{A_\alpha}(x) \textbf{1}_{\bar A_\alpha}(y)\\
= 2\int_{I_c}\rmd \alpha \int_{A_\alpha}\pi(\rmd x)
K(x,\bar A_\alpha)\\
\geq 2\Phi_{\epsilon'}(K) \int_{I_c} \left(\pi(A_\alpha)-\epsilon'\right)\left(\pi(\bar A_\alpha)-\epsilon'\right)\rmd \alpha \\
= 2\Phi_{\epsilon'}(K) \int_{I_c}\left(\pi(A_\alpha)\pi(\bar A_\alpha) -\epsilon'(1-\epsilon')\right) \rmd \alpha,
\end{multline*}
where $\epsilon'\eqdef \epsilon/32$.
\begin{multline*}
\int_{I_c} \pi(A_\alpha)\pi(\bar A_\alpha)\rmd\alpha  = \int_{0}^{\infty}  \pi(A_\alpha)\pi(\bar A_\alpha)\rmd\alpha \\
=\int (f_c^2(y) - f_c^2(x))\textbf{1}_{\{(x,y): f_c^2(y)> f_c^2(x)\}} \pi(\rmd x)\pi(\rmd y)\\
=\frac{1}{2}\int \abs{f_c^2(y) - f_c^2(x)}\pi(\rmd x)\pi(\rmd y),
\end{multline*}
whereas
\[\int_{I_c}\epsilon'(1-\epsilon') \rmd \alpha = \epsilon'(1-\epsilon')\left(\texttt{ess-sup}(f_c^2) - \texttt{ess-inf}(f_c^2)\right).\]
For $y,z\in\Xset$,
\[|f_c(y)|^2 - |f_c(z)|^2 \leq \left|f(y) - f(z)\right|\left|f_c(y) + f_c(z)\right| \leq 4\|f\|_\infty\;\|f_c\|_\infty,\]
which implies that
\[\int_{I_c}\epsilon'(1-\epsilon') \rmd \alpha \leq 4\epsilon' \|f_c\|_\infty.\]
We conclude that
\begin{equation}\label{proof:cheeger:eq:1}
\sqrt{\e_K(f,f)} \geq \frac{\Phi_{\epsilon'}(K)}{\sqrt{8\pi(f_c^2)}}  \left(\int \abs{f_c^2(y) - f_c^2(x)}\pi(\rmd x)\pi(\rmd y) - 8\epsilon' \|f_c\|_\infty\right).\end{equation}
Since (\ref{proof:cheeger:eq:1})  holds true for all $c$, we get, by dominated convergence that
\begin{multline*}
\sqrt{\e_K(f,f)} \geq \limsup_{c\to+\infty} \frac{\Phi_{\epsilon'}(K)}{\sqrt{8\pi(f_c^2)}}  \left(\int \abs{f_c^2(y) - f_c^2(x)}\pi(\rmd x)\pi(\rmd y) - 8\epsilon' \|f_c\|_\infty\right)\\
= \frac{\Phi_{\epsilon'}(K)}{\sqrt{8}}\left(2\int \abs{f(y) - f(x)}\pi(\rmd x)\pi(\rmd y) -8\epsilon'\right).
\end{multline*}
By Jensen's inequality
\[\int\int|f(y)-f(x)|\pi(\rmd x)\pi(\rmd y)\geq \int|f(x) -\pi(f)|\pi(\rmd x) =\pi(|\bar f|),\]
where $\bar f\eqdef f-\pi(f)$. If $\pi(|\bar f|) \geq \sqrt{\textsf{Var}_\pi(f)}/2$, we conclude that
\begin{equation*}
\sqrt{\e_K(f,f)} \geq  \frac{\Phi_{\epsilon'}(K)}{\sqrt{8}}\left(\sqrt{\textsf{Var}_\pi(f)} - 8\epsilon'\right).\end{equation*}
Since the function $f$ is taken such that  $\textsf{Var}_\pi(f) >\epsilon$, using the  convexity inequality $\sqrt{x+y} \geq (\sqrt{x}+\sqrt{y})/\sqrt{2}$ valid for all $x,y\geq 0$ we have
\begin{multline*}
\sqrt{\textsf{Var}_\pi(f)} - 8\epsilon' = \sqrt{\textsf{Var}_\pi(f) - \frac{\epsilon}{2} + \frac{\epsilon}{2}} - \frac{\epsilon}{4}\geq \frac{\sqrt{\textsf{Var}_\pi(f) - \frac{\epsilon}{2}}}{\sqrt{2}} + \frac{\epsilon^{1/2}}{2} - \frac{\epsilon}{4}\\
\geq \frac{\sqrt{\textsf{Var}_\pi(f) - \frac{\epsilon}{2}}}{\sqrt{2}}.\end{multline*}
It follows that if $\pi(|\bar f|) \geq \sqrt{\textsf{Var}_\pi(f)}/2$, we have
\begin{equation}\label{proof:cheeger:eq:2}
\sqrt{\e_K(f,f)} \geq  \frac{\Phi_{\epsilon'}(K)}{\sqrt{16}}\sqrt{\textsf{Var}_\pi(f) - \frac{\epsilon}{2}}.\end{equation}

Going back to (\ref{proof:cheeger:eq:1}), using the value at $c=\pi(f)$, and recalling the notation $\bar f = f-\pi(f)$, we have
\begin{equation}\label{proof:cheeger:eq:3}
\sqrt{\e_K(f,f)} \geq \frac{\Phi_{\epsilon'}(K)}{\sqrt{8\pi(\bar f^2)}}   \left(\int\int \abs{\bar f^2(y) - \bar f^2(x)}\pi(\rmd x)\pi(\rmd y) - 16\epsilon'\right).
\end{equation}
We have
\[\int \abs{\bar f^2(y) - \bar f^2(x)}\pi(\rmd y)\geq \left|\int (\bar f^2(y) - \bar f^2(x))\pi(\rmd y)\right| =\left|\bar f^2(x) - \textsf{Var}_\pi(f)\right|,\]
so that
\begin{multline*}
\int\int \abs{\bar f^2(y) - \bar f^2(x)}\pi(\rmd x)\pi(\rmd y) \geq \int \left|\bar f^2(x) - \textsf{Var}_\pi(f)\right|\pi(\rmd x)\\
=\int\left\{\bar f^2(x) + \textsf{Var}_\pi(f) -2\min\left(\bar f^2(x),\textsf{Var}_\pi(f)\right)\right\}\pi(\rmd x)\\
=2\textsf{Var}_\pi(f) -2\textsf{Var}_\pi(f)\int\min\left(\left(\frac{\bar f(x)}{\sqrt{\textsf{Var}_\pi(f)}}\right)^2,1\right)\pi(\rmd x)\\
\geq 2\textsf{Var}_\pi(f) -2\textsf{Var}_\pi(f)\frac{\pi(|\bar f|)}{\sqrt{\textsf{Var}_\pi(f)}},
\end{multline*}
where the last inequality in the last display follows from the observation that $\min(1,x^2)\leq |x|$ for all $x\in\rset$. Hence if $\pi(|\bar f|)\leq \sqrt{\textsf{Var}_\pi(f)}/2$ we have
\[\int\int \abs{f^2(y) - f^2(x)}\pi(\rmd x)\pi(\rmd y)  \geq \textsf{Var}_\pi(f),\]
and we use this with (\ref{proof:cheeger:eq:3}) to deduce that
\[\sqrt{\e_K(f,f)} \geq \frac{\Phi_{\epsilon'}(K)}{\sqrt{8}}\times \left(\sqrt{\textsf{Var}_\pi(f) } -\frac{ 16\epsilon'}{\sqrt{\textsf{Var}_\pi(f)}}\right).\]
Since  $\textsf{Var}_\pi(f) >\epsilon$,
\begin{multline*}
\sqrt{\textsf{Var}_\pi(f) } -\frac{ 16\epsilon'}{\sqrt{\textsf{Var}_\pi(f) }} \geq \sqrt{\textsf{Var}_\pi(f) - \frac{\epsilon}{2} + \frac{\epsilon}{2}} -\frac{\epsilon^{1/2}}{2} \\
\geq \frac{\sqrt{\textsf{Var}_\pi(f)  -\frac{\epsilon}{2}}}{\sqrt{2}} + \frac{\epsilon^{1/2}}{2} -\frac{\epsilon^{1/2}}{2}
\geq \frac{\sqrt{\textsf{Var}_\pi(f)  -\frac{\epsilon}{2}}}{\sqrt{2}}.
\end{multline*}
In conclusion, if $\textsf{Var}_\pi(f)  > \epsilon$ we have
\begin{equation}\label{proof:cheeger:eq:4}
\sqrt{\e_K(f,f)} \geq \frac{\Phi_{\epsilon'}(K)}{\sqrt{16}}\left(\sqrt{\pi(f^2) - \epsilon/2}\right).
\end{equation}
We combine the last display and (\ref{proof:cheeger:eq:2}) to conclude that for all $f\in L^2(\pi)$ such that $\|f\|_\infty\leq 1$, and $\textsf{Var}_\pi(f) >\epsilon$, we have
\[\frac{\e_K(f,f)}{\pi(f^2) - \epsilon/2} \geq \frac{\Phi_{\epsilon'}^2(K)}{16},\]
which implies the lower bound.

\subsection{Proof of Proposition \ref{lb:cp}}\label{sec:proof:proplbcp}
Take a function $f:\;\Xset\to\rset$ such that $\|f\|_\infty \leq 1$, and $\textsf{Var}_\pi(f)>\zeta$. We have
\begin{multline*}
2\textsf{Var}_\pi(f) = \sum_{x\in\Xset_0}\sum_{y\in\Xset_0}(f(y)-f(x))^2\pi(x)\pi(y)  + 2\sum_{x\in\Xset_0}\sum_{y\in\Xset\setminus\Xset_0}(f(y)-f(x))^2\pi(x)\pi(y)\\
+ \sum_{x\in\Xset\setminus \Xset_0}\sum_{x\in\Xset\setminus \Xset_0}(f(y)-f(x))^2\pi(x)\pi(y)\\
\leq\sum_{x\in\Xset_0}\sum_{y\in\Xset_0}(f(y)-f(x))^2\pi(x)\pi(y) + 8\pi(\Xset\setminus\Xset_0). \end{multline*}
Using $\pi(\Xset_0)\geq 1 - (\zeta/8)$, we get 
\[2\left(\textsf{Var}_\pi(f)-\frac{\zeta}{2}\right) \leq \sum_{x\in\Xset_0}\sum_{y\in\Xset_0}(f(y)-f(x))^2\pi(x)\pi(y).\]
Hence
\[\frac{\e(f,f)}{\textsf{Var}_\pi(f)-\frac{\zeta}{2}} \geq \frac{\sum_{x\in\Xset_0}\sum_{y\in\Xset_0}(f(y)-f(x))^2\pi(x)K(x,y)}{\sum_{x\in\Xset_0}\sum_{y\in\Xset_0}(f(y)-f(x))^2\pi(x)\pi(y)}.\]
For $x\neq y$ in $\Xset_0$, let $\gamma_{xy}$ be a canonical path linking $x,y$ so that we can write, using the Cauchy-Schwarz inequality:
\begin{multline*}
(f(y)- f(x))^2 = \left(\sum_{e\in \gamma_{xy}}f(e_-) - f(e_+)\right)^2 \\
\leq \sum_{e\in \gamma_{xy}}\; \frac{|\gamma_{xy}|}{\pi(e_-)K(e_-,e_+)}\times  \pi(e_-)K(e_-,e_+)(f(e_-) - f(e_+))^2.\end{multline*}
It follows that
\begin{multline*}
\sum_{x\neq y}\left(f(y)- f(x)\right)^2\pi(x) \pi(y)\\
\leq \sum_{x\neq y} \sum_{e\in \gamma_{xy}}\frac{\pi(x)\pi(y)|\gamma_{xy}|}{\pi(e_-)K(e_-,e_+)}  \left\{\pi(e_-)K(e_-,e_+)(f(e_-) - f(e_+))^2\right\}\\
= \sum_{e\in\e_0} \;\sum_{\gamma_{xy}\ni e} \;\left\{\frac{\pi(x)\pi(y)|\gamma_{xy}|}{\pi(e_-)K(e_-,e_+)}\right\} \left\{\pi(e_-)K(e_-,e_+)(f(e_-) - f(e_+))^2\right\} \\
\leq \mathsf{m}(\Xset_0) \times \sum_{e\in\e_0} \pi(e_-)K(e_-,e_+)(f(e_-) - f(e_+))^2 \\
\leq \mathsf{m}(\Xset_0) \times \sum_{x\in\Xset_0}\sum_{y\in\Xset_0}\left(f(y)- f(x)\right)^2\pi(x) K(x,y).
\end{multline*}
As a result we conclude that for all $f:\;\Xset\to\rset$, with $\|f\|_\infty\leq 1$,  and $\textsf{Var}_\pi(f)>\zeta$
\[\frac{\e(f,f)}{\textsf{Var}_\pi(f)-\frac{\zeta}{2}} \geq \frac{1}{\mathsf{m}(\Xset_0)}.\]

\bibliographystyle{ims}
\bibliography{document}

\end{document}

%% file: header_main.tex
\usepackage{amsmath,amssymb,mathtools,amsfonts,amsthm,enumerate,natbib,color,ifthen,hyperref}
\usepackage{xargs}
\usepackage[textwidth=4cm, textsize=footnotesize]{todonotes}
\usepackage[latin1]{inputenc}
\usepackage{wrapfig}

\usepackage{float}
\usepackage[caption = false]{subfig}
\usepackage{graphicx}

\usepackage{aliascnt,bbm}
\def\rset{\mathbb R}
\def\zset{\mathbb Z}

\def\eqsp{\;}

\newcommand{\pscal}[2]{\left\langle#1,#2\right\rangle}

\newcommand{\eqdef}{\ensuremath{\stackrel{\mathrm{def}}{=}}}

\def\Xset{\mathcal{X}} 
\def\B{\mathcal{B}} 
\def\e{\mathcal{E}}

\def\m{\mathsf{m}}

\def\G{\mathcal{G}}
\def\D{\mathcal{D}}
\def\A{\mathcal{A}}

\def\H{\mathcal{H}}

\newcommandx\sequence[3][2=t,3=\zset]
{\ifthenelse{\equal{#3}{}}{\ensuremath{\{ #1_{#2}\}}}{\ensuremath{\{ #1_{#2}, \eqsp #2 \in #3 \}}}}

\def\PP{\mathbb{P}} 
\newcommand{\CPP}[3][]
{\ifthenelse{\equal{#1}{}}{{\mathbb P}\left(\left. #2 \, \right| #3 \right)}{{\mathbb P}_{#1}\left(\left. #2 \, \right | #3 \right)}}
\def\PE{\mathbb{E}} 
\newcommand{\CPE}[3][]
{\ifthenelse{\equal{#1}{}}{{\mathbb E}\left[\left. #2 \, \right| #3 \right]}{{\mathbb E}_{#1}\left[\left. #2 \, \right | #3 \right]}}

\def\wtilde{\widetilde} 

\def\tv{\mathrm{tv}}

\def\Cset{\mathcal{C}} 

\newcommand{\abs}[1]{\left\vert#1\right\vert}

\usepackage{bm}

\theoremstyle{plain}
\newtheorem{theorem}{Theorem}
\newtheorem{assumption}{H\hspace{-3pt}}

\newaliascnt{proposition}{theorem}
\newtheorem{proposition}[proposition]{Proposition}
\aliascntresetthe{proposition}

\newaliascnt{lemma}{theorem}
\newtheorem{lemma}[lemma]{Lemma}
\aliascntresetthe{lemma}
\newaliascnt{corollary}{theorem}

\aliascntresetthe{corollary}

\theoremstyle{definition}
\newaliascnt{definition}{theorem}

\aliascntresetthe{definition}

\newtheorem{algorithm}{Algorithm}

\newaliascnt{remark}{theorem}
\newtheorem{remark}[remark]{Remark}
\aliascntresetthe{remark}

\newaliascnt{example}{theorem}

\aliascntresetthe{example}

\def\rmd{\mathrm{d}}

\def\1{\mathbbm{1}}

\DeclareMathOperator*{\argmin}{{\mathsf{Argmin}}}
\DeclareMathOperator*{\argmax}{{\mathsf{Argmax}}}